\documentclass{article}

\usepackage{style}
\usepackage[margin=1in]{geometry}

\begin{document}

\title{Learning-Based Algorithms for Graph Searching Problems}

\author{Adela Frances DePavia\thanks{University of Chicago. Email: \texttt{adepavia@uchicago.edu}. Supported by NSF DGE 2140001.}
\and
Erasmo Tani\thanks{University of Chicago. Email: \texttt{etani@uchicago.edu}. Partially supported by the Institute for Data, Econometrics, Algorithms, and Learning (IDEAL) with NSF Grant ECCS-2216912.}
\and
Ali Vakilian\thanks{Toyota Technological Institute at Chicago. Email: \texttt{vakilian@ttic.edu}.}
}
\maketitle
\begin{abstract}
    We consider the problem of graph searching with prediction recently introduced by \cite{banerjee2022graph}. In this problem, an agent, starting at some vertex $r$ has to traverse a (potentially unknown) graph $G$ to find a hidden goal node $g$ while minimizing the total distance travelled. We study a setting in which at any node $v$, the agent receives a noisy estimate of the distance from $v$ to $g$. We design algorithms for this search task on unknown graphs. We establish the first formal guarantees on unknown weighted graphs and provide lower bounds showing that the algorithms we propose have optimal or nearly-optimal dependence on the prediction error. Further, we perform numerical experiments demonstrating that in addition to being robust to adversarial error, our algorithms perform well in typical instances in which the error is stochastic. Finally, we provide alternative simpler performance bounds on the algorithms of \citet{banerjee2022graph} for the case of searching on a known graph, and establish new lower bounds for this setting.
\end{abstract}

\section{Introduction}

Searching on graphs is a fundamental problem which models many real-world applications in autonomous navigation. In a \textit{graph searching problem} instance, an agent is initialized at some vertex $r\in V$ (referred to as the root) in some (potentially weighted, directed) graph $G = (V,E)$. The agent's task is to find a goal node $g \in V$. The agent searches for $g$ by sequentially visiting adjacent nodes in the graph. The graph searching problem terminates when the agent reaches the goal node, and the cost incurred by the agent is the total amount of distance they travelled.

There are two main settings of interest. In the \textit{exploration} setting, as the agent moves through the graph, it only learns the structure of $G$ by observing the vertices and edges adjacent to the nodes it has visited, a model sometimes referred to as the \emph{fixed graph scenario} \citep{komm2015treasure, kalyanasundaram1994constructing}. In the strictly-easier \textit{planning} setting, the agent is given the entire graph $G$ ahead of time, but does not know the identity of the goal node $g$.

Without additional information, in the worst-case the agent must resort to visiting the entire graph, a task  which amounts to finding an efficient tour in an unknown graph\footnote{We note that historically the task of finding an efficient tour in an unknown graph has been referred to as graph exploration, but following the conventions of \citet{banerjee2022graph} we reserve this name for the graph search problem on an unknown graph.}\citep{berman2005line, dobrev2012online, megow2012online, eberle2022robustification}. Recently, \citet{banerjee2022graph} consider the setting when an algorithm for graph searching also receives some prediction function, $f:V\rightarrow \R$, representing some (noisy) estimate of the distance to the goal node $g$ at any given node $v\in V$. This setup models applications in which the searcher receives advice from some machine-learning model designed to predict the distance to the goal; this problem fits into the broader framework of learning-based algorithms, which exploit (potentially noisy) advice from a machine learning model to enhance their performance. %

In the case of planning, \citet{banerjee2022graph} propose an intuitive strategy that can be deployed in weighted and unweighted graphs and analyze its performance in terms of structural properties of the instance graph, such as its maximum-degree and its doubling-dimension. They establish formal guarantees on the cost incurred by their algorithm under different notions of prediction error. In contrast, their results in the exploration setting are limited: they propose an algorithm for exploration in unweighted trees. Their algorithm is tailored to this restricted class of graphs and their guarantees are parameterized by the number of incorrect predictions, which does not capture the magnitude of the deviation between predictions and true distances.

This paper seeks to expand the understanding of graph exploration problems with predictions. We design algorithms which can be deployed on a variety of weighted graphs and prove worst-case guarantees on their performance. We also complement this analysis by providing lower-bounds for these problems, showing that the algorithms studied in this work are optimal or nearly optimal. In particular, we focus on two different error models. In the \textbf{absolute} error model, the magnitude of the error at a node is independent of that node's true distance to the goal, and guarantees are given in terms of the total magnitude of the error incurred at every node. In the \textbf{relative} error model, nodes further from the goal may have larger deviation between the prediction and the truth, and algorithmic guarantees are parameterized by the maximum ratio of the error to the true distance at any vertex.

\paragraph{Related Works}  Online graph searching problems have long been used as basic models for problems in autonomous navigation \cite{berman2005line}. Searching with access to predictions, also referred to as ``advice'' or ``heuristics'' in different communities, is a commonly studied variant \citep{pelc2002searching, dobrev2012online, eberle2022robustification, banerjee2022graph}. The problem and prediction settings considered in this work most closely correspond to those considered by \citet{banerjee2022graph}. This setup models applications in which predictions are the output of some machine-learning model. Recent years have seen a marked increase  in the integration of machine learning techniques to enhance traditional algorithmic challenges \citep{angelopoulos2020online,gupta2022augmenting, mitzenmacher2022algorithms, antoniadis2023online}. More generic forms of advice and the advice-complexity of exploration tasks are long-standing subjects of study. \citet{komm2015treasure} study the case when the the searcher receives generic advice, which can take the form of any bit string, and prove results about the advice complexity of this task. %
For a more detailed survey of related works, we direct the reader to Section~\ref{ssec:related_work}.

\paragraph{Organization} In Section~\ref{ssec:results} we formally state the main results of the paper. In Section~\ref{sec:preliminaries} we present technical preliminaries and define relevant notation. Sections~\ref{sec:explo_L1_error} and \ref{sec:multiplicative_error} contain algorithms and analysis for exploration under absolute and relative error models respectively. In Section~\ref{sec:planning} we derive new bounds in the planning setting via metric embeddings. Finally, in Section~\ref{sec:experiments} we complement these results with numerical experiments. All missing proofs are in Section~\ref{sec:missing_proofs} of the supplementary material.

\subsection{Summary Of Results}\label{ssec:results}

We begin by formally describing the exploration and planning settings:
\begin{description}
    \item[ {The Exploration Problem}] In this setting, both the graph $G$ and the predictions $f$ are initially unknown to the agent: the agent is initialized with access to the root node $r$, the neighbors of $r$, and the predictions at all of these nodes. As the algorithm proceeds, on each iteration $i$ it has access in memory to a subgraph $G_i \subseteq G$ containing the nodes it has visited, the neighbors of those nodes, and any edges between visited nodes and neighbors. The searcher can only query predictions from nodes in the subgraph $G_i$. This problem models exploration of an unknown environment.
    \item[ {The Planning Problem}] In this setting, both the graph $G$ and the predictions $f$ are fully known to the searcher upon initialization. 
\end{description}

For any algorithm which visits an ordered sequence of vertices $v_1,\dots,v_T$, we denote the algorithmic cost $\alg \defeq \sum_{i\in T} d_{G_i}(v_{i}, v_{i+1})$. The guarantees in this paper compare $\alg$ to the optimal cost a posteriori, denoted by $\opt \defeq d_{G}(r,g)$.

\paragraph{Exploration Under Absolute Error} A natural way of measuring the error of some predictions is the magnitude of the difference between the true distance-to-goal and prediction value at each node. In Section~\ref{sec:explo_L1_error} we propose an algorithm for the exploration problem on weighted graphs and prove performance guarantees parameterized by these error measures. In  particular, we prove the following theorem:  
\begin{restatable}{theorem}{ell1greedybound}\label{thm:searching-graphs-greedy}
    There is an algorithm for searching arbitrary (potentially directed) graphs which finds the goal $g$ by traveling a distance of at most
        $\operatorname{OPT} + \,\cE_1^-  + n\cdot \cE_\infty^+$,
    where $\cE_1^- \defeq \sum_{v \in V} \max \left\{0, d(v,g)-f(v)\right\}$ and $\cE_\infty^+ \defeq \max_{v\in V} \max\left\{0, f(v) - d(v,g)\right\}$.
\end{restatable}

This algorithm enjoys several advantages over the most recent results on the exploration problem by \citet{banerjee2022graph}. Their work introduces an involved combinatorial algorithm for exploration on unweighted trees whose performance is parametrized by the $\ell_0$ norm of the vector of errors. The guarantees of their algorithm  do not apply when the graph being searched is not an unweighted tree. In contrast, the algorithm proposed in the present work is intuitive and easily implementable, and the guarantees obtained hold in a wide variety of settings, e.g. when the graph is weighted and/or directed. The performance of the algorithm is parameterized by the natural absolute deviation ($\ell_1$) error measure.

We also establish that under the above parameterization, the proposed algorithm is in some sense optimal (see Theorem~\ref{thm:optimal_E1_lowerbounds}). Further our analysis highlights that the same algorithm performs particularly well when (erroneous) predictions only \textit{under}estimate distance to the true goal. Prediction functions with this property are referred to as \textit{admissible}. They are key objects of study in path-finding literature and are well-motivated by applications~\citep{dechter1985generalized,eden2022embeddings, ferguson2005guide,pohl1969bi}.

\paragraph{Relative Errors} 

One realistic setting for applications is one in which the error is proportionate to the magnitude of the distance to the goal. In order to capture this behavior, we consider a different error model in which the ratio of the error to true distance-to-goal at every vertex is assumed to be bounded by some value $\varepsilon \in (0,1)$:
\begin{equation}\label{eq:multiplicative_error_setup}
    (1-\varepsilon) d(v,g) \leq f(v)\leq (1+\varepsilon) d(v,g).
\end{equation}
We do not place any restriction on the total amount of error in the graph beyond this condition. 

We consider two regimes of multiplicative error. In the first setting, $\varepsilon$ is assumed to be known to the searcher a priori. We propose an algorithm and show that it achieves the following competitive ratio:
 \begin{theorem}\label{cor:multiplicative-error-trees}
    Consider the exploration problem on a weighted tree where predictions satisfy \eqref{eq:multiplicative_error_setup} with respect to $\varepsilon \in (0,1)$, and $\varepsilon$ is known. Then there exists an algorithm which succeeds in finding the goal $g$ and incurs competitive ratio at most
    \[
        \frac{\alg}{\opt} \leq \frac{1}{1-\varepsilon} + n\varepsilon \cdot \frac{4}{(1-\varepsilon)^2}.
    \]

    In particular, if the predictions are admissible, then the same algorithm  incurs competitive ratio
    \[
        \frac{\alg}{\opt} \leq 1 + n\varepsilon \cdot \frac{2}{1-\varepsilon}.
    \]
\end{theorem}

In the second regime $\varepsilon$ is assumed to be small ($\varepsilon < 1/3$) but its exact value is not assumed to be known. For this setting, we design a different algorithm which allows us to prove the following result:
\begin{theorem}\label{thm:exploration-trees-epsilon-unknown}
    Given $G$ a weighted tree with predictions $f$ satisfying Equation~\eqref{eq:multiplicative_error_setup} for some unknown $\varepsilon<{1/3}$, then Algorithm~\ref{alg:weighted-vareps-unknown-search} with $\beta = 2/3$ incurs competitive ratio at most
    \[
        {\alg\over \opt} \leq 2 + O\left(n\varepsilon\frac{5+3\varepsilon}{(1-3\varepsilon)^2}\right).%
    \]
\end{theorem}
In Section~\ref{sec:multiplicative_error},  we describe our algorithms for these problems, prove the above theorems, and complement these algorithmic guarantees with lowerbounds that show these algorithms are nearly optimal.

\paragraph{Planning Problems}
\citet{banerjee2022graph} consider problems in which both the full graph and all predictions are available to the algorithm upon initialization. This setting is referred to as the planning problem. They construct algorithms for this version of the problem  under different error models and the guarantees they obtained are outlined in Table~\ref{table:planning}. In many regimes (e.g. planning on unweighted trees under error parametrized by the $\ell_0$-norm of the vector of errors, denoted $\cE_0$, and planning on graphs under error parametrized by the $\ell_1$-norm of the vector of errors, denoted $\cE_1$) matching lowerbounds for their algorithms can be established, as shown in the table. A notable exception is the case of planning on unweighted graphs under $\cE_0$ parametrization: they establish an upperbound of $\opt + 2^{O(\alpha)} O(\cE_0^2)$ where $\alpha$ is the \textit{doubling dimension} of the graph. In particular, the lowerbounds they provide fail to match the depedence on the quadratic term $\cE_0^2$ in their upperbound.

We provide an alternative analysis of their algorithm based on metric properties of the instance graph which shows that in some classes of graphs one can reduce the asymptotic dependence on $\cE_0$ from a quadratic to a linear factor. In particular, we consider the distortion of embedding the instance graph into a weighted path or cycle graph, and establish the following guarantee:
\begin{restatable}{theorem}{planningembeddingE0}\label{theorem:planning-path-embedding-E0}
    Consider $G$ an unweighted graph such that $G$ admits an embedding into a weighted path or a weighted cycle of distortion at most $\rho$. Then on $G$, the $\cE_0$ planning algorithm of \citet{banerjee2022graph} incurs cost at most
        $\opt + O(\rho \cE_0)$.
\end{restatable}

We note that the results present in Table~\ref{table:planning} which were not proved by \citet{banerjee2022graph} are established in this paper in Section~\ref{sec:planning} and the proved in the supplementary material.

\begin{table*}[]
\begin{center}
\def\arraystretch{1.75}%
\begin{tabular}{cc cc}
&   & $\cE_0$, unweighted & $\cE_1$, positive weights \\ \hline\hline                                               
\multicolumn{1}{l||}{\multirow{2}{*}{\begin{tabular}[x]{@{}c@{}}Trees with\vspace{-3mm}\\ maximum degree $\Delta$ \end{tabular}}}  & \small{upper bound} &   $\opt + O(\Delta \cE_0)$ ({\dag}) & \begin{tabular}[x]{@{}c@{}} $\opt + O(\Delta \cE^2_1)$\vspace{-4mm}\\\small(integer distances)\end{tabular}    \\ %
\multicolumn{1}{l||}{}                        & \small{lower bound} &   $\max\left\{\opt, \Omega(\Delta \cE_0)\right\}$ {(\dag)} & $\max\left\{\opt,\Omega(\Delta\cE^2_1)\right\}$
\\ \hline
\multicolumn{1}{l||}{\multirow{2}{*}{\begin{tabular}[x]{@{}c@{}}Graphs with\vspace{-3mm}\\ doubling dimension $\alpha$ \end{tabular}}} & \small{upper bound} & $\opt + 2^{O(\alpha)}O(\cE_0^2)$ {(\dag)}  &  $\opt + 2^{O(\alpha)}O(\cE_1)$ {(\dag)}  \\ %
\multicolumn{1}{l||}{}                        & \small{lower bound} &  Unknown%
&  $\max\left\{\opt,\Omega(2^{\alpha}\cE_1)\right\}$
\\ \hline
\multicolumn{1}{l||}{\multirow{2}{*}{\begin{tabular}[x]{@{}c@{}}Graphs with path-\vspace{-3mm}\\ embedding distortion $\rho$ \end{tabular}}} & \small{upper bound} & $\opt + O(\rho\cE_0)$  &  $\opt + O(\rho\cE_1)$  \\ %
\multicolumn{1}{l||}{}                        & \small{lower bound} &  Unknown  &  $\max\left\{\opt, \Omega(\rho\cE_1)\right\}$
\end{tabular}
\caption{Known results for planning problem in different settings. {(\dag)} denotes results from \citet{banerjee2022graph}. $\Delta$ denotes the maximum degree of any vertex in the graph, $\alpha$ denotes the doubling-dimension of the graph, and $\rho$ denotes the distortion of the path-embedding on $G$. \citet{banerjee2022graph} in their work note the absence of a matching lowerbound in the graph planning problem. This gap motivates this work's study of planning bounds parameterized by metric embeddings, which yields a reduced asymptotic dependency on $\cE_0$. For full discussions of lowerbounds for $\cE_1$ parametrized by $\Delta$, $\alpha$, and $\rho$, see Lemmas~\ref{lemma:lowerbound-E1-planning} and \ref{lemma:lowerbound-E1-planning-trees}, and Lemma~\ref{lemma:lowerbound-E1-planning-rho} in Section~\ref{sec:missing_proofs} of the supplementary material.}\label{table:planning}
\end{center}
\end{table*}

\section{Technical Preliminaries and Notation}\label{sec:preliminaries}

Graphs are assumed to be weighted and directed, unless otherwise specified. (Weighted) shortest-path distances in a graph $G$ are denoted by $d_{G}(\cdot, \cdot)$. Given a set $S \subseteq V$, let $\partial S$ be its external vertex boundary:
\[
    \partial S \defeq \{v\in G\setminus S \mid \exists u \in S: v\sim u\}.
\]
For a vertex set $S \subseteq V$, we denote by $\tour_{G}(S)$ the (weighted) length of the shortest walk that visits all nodes in $S$:
\begin{equation}\label{eq:defn_tour}
    \tour_{G}(S) \defeq \max_{v\in S}\min_{W\in\mathcal{W}(v,S)} \text{length}_G(W),
\end{equation}
where $\mathcal{W}(v,S)$ is the set of walks in $G$ starting at vertex $v$ and visiting every vertex in $S$.

Given a metric space $(X,d_X)$ its doubling constant is the minimum number $\lambda$ such that, for every $r > 0$, every ball of radius $r$ can be covered by at most $\lambda$ balls of radius $r/2$ \citep{gupta2003bounded}. The doubling constant of a graph $G$ is the doubling constant of $(G,d)$ where $d$ is the shortest path distance on $G$. The doubling dimension of the space, denoted $\alpha$, is defined as $\alpha \defeq \log_2{\lambda}$. 

\paragraph{Notation For Exploration Algorithms} Recall that in exploration problems, the true graph $G$ is initially unknown to the searcher. Thus in the setting of exploration, one needs to distinguish between the shortest \emph{known} path distance between two vertices and the true shortest path distance. To this end, let $V_{i-1}$ be the set of vertices \emph{visited} by iteration $i$ and let $G_i$ be the subgraph of $G$ containing: all of the vertices in $V_{i-1}\cup \partial V_{i-1}$, and all of the edges adjacent to $V_{i-1}$. Throughout this paper, we emphasize the distinction between $d_{G_i}$ and $d_G$. In general, we have $d_{G_i} \geq d_G$.

When analyzing performance, we denote the \textit{progress} made on the $i$th iteration as
\begin{equation}\label{eq:defn-Delta}
    \Delta_i \defeq d_{G}(v_i, g) - d_{G}(v_{i+1},g).
\end{equation}
Observe that, under the assumption that $v_0, v_1,\dots, v_T$ are nodes visited by some algorithm which originates at $r$ and terminates at $g$ (i.e. $v_0 = r$ and $v_T = g$) we have:
\[
    \sum_{i=0}^{T-1} \Delta_i = d_{G}(r,g) - 0 = \opt.
\]

\paragraph{Metric Embeddings} When analyzing the planning algorithms of \citet{banerjee2022graph}, we consider metric embeddings on graphs, and parametrize results in terms of the \textit{distortion} of the relevant embedding:
\begin{definition}[Distortion]
    Given a function $\tau : X \to Y$ between two finite metric spaces $(X, d_X)$ and $(Y,d_Y)$, we define the distortion $\operatorname{dist}(\tau)$ of $\tau$ as the minimum value $\rho$ satisfying the following: there exists a constant $c > 0$ such that for all $x_1, x_2 \in X$,
    \[
        c \cdot d_X(x_1,x_2) \leq d_Y(\tau(x_1), \tau(x_2)) \leq  c \cdot \rho \cdot d_X(x_1,x_2).
    \]
\end{definition}

\section{Exploration Under Absolute Error}\label{sec:explo_L1_error}

In this section we study the absolute error regime, in which the $\ell_1$ norm of the vector of errors is bounded by some constant $\cE_1$, not necessarily known to the searcher. We consider the following natural rule: on the $i$th iteration, choose the next vertex to visit by picking the node $ v_i\in \partial V_{i-1}$ minimizing the sum of $d_{G_i}(v_{i-1},v_i)$ and $f(v_i)$ (See Algorithm~\ref{alg:greedy-l1-search}). This iterative step can be interpreted as visiting the vertex that would be on the shortest path to the goal if all the predictions $f(v)$ were correct. 

\begin{algorithm}[H]
    \begin{algorithmic}
    \caption{\basicalgo $(G,r)$}\label{alg:greedy-l1-search}
        \State $v_0 \gets r$
        \State $i \gets 0$
        \State $V_0 \gets \{r\}$
        \While{$v_i \neq g$}
            \State $i\gets i+1$
            \State $v_i \in \argmin_{v\in \partial V_{i-1}} d_{G_i}(v_{i-1},v) + f(v)$
            \State $V_i \gets \{v_0, ... , v_i\}$
        \EndWhile
    \end{algorithmic}
\end{algorithm}

We remark that because we are working in the exploration setting, Algorithm~\ref{alg:greedy-l1-search} does not have access to the true distances $d_G(\cdot, \cdot)$ and must instead make use of $d_{G_i}(\cdot, \cdot)$ the distances in the subgraph of observed vertices at iteration $i$. 

Theorem~\ref{thm:searching-graphs-greedy} parametrizes the worst-case guarantees for Algorithm~\ref{alg:greedy-l1-search} in terms of the total negative error $\cE^{-}$ and the maximum-occurring positive error $\cE^{+}_{\infty}$. This asymmetry corresponds to the intuition that positive errors can obstruct the search task more dramatically than negative errors by obscuring shortest paths to the goal. Indeed, an immediate corollary of Theorem~\ref{thm:searching-graphs-greedy} is the following result about the performance of Algorithm~\ref{alg:greedy-l1-search} in the setting in which the prediction function is admissible (i.e. error is only negative):
\begin{corollary}\label{cor:decremental-error-search}
    Consider the problem of searching a weighted (possibly directed) graph with predictions $f$ satisfying
        $f(v) \leq d_{G}(v,g) \ \forall v\in V$
    Then there exists an algorithm which finds the goal $g$ with cost at most
        $\operatorname{OPT} + \,\cE_1$,
    where $\cE_1$ is the total $\ell_1$ error in the predictions, i.e. $\cE_1\defeq \sum_{v\in V} |f(v) - d(v,g)|$.
\end{corollary}

The dependency on $\cE^{-}$ and $\cE^{+}_{\infty}$ in Theorem~\ref{thm:searching-graphs-greedy} is optimal in the following sense: 

\begin{theorem}\label{thm:optimal_E1_lowerbounds}
    For every $\cE^{-} > 0$, there exist graph search instances with total negative error $\cE^{-}$ such that any algorithm for the exploration problem on these instances must incur cost at least $\opt + \cE^{-}$ in the worst case. Additionally, for any $n > 3$ and any $\cE^{+}_{\infty}$, there exist graph search instances on $n$ nodes with maximum positive error $\cE^{+}_{\infty}$ such that any algorithm for the exploration problem on these instances must incur cost at least $\opt + \cE^{+}_{\infty}(n-2)$ in the worst case.
\end{theorem}

We note that the lower bound in \Cref{thm:optimal_E1_lowerbounds} also holds for the expected distance travelled of randomized search strategies, up to constant factors, as per the following result.

\begin{proposition}\label{lem:lower-bounds-randomized}
    For every $\cE^{-}>0$ there exists a graph search instance with total negative error $\cE^{-}$ such that any randomized algorithm incurs expected costs at least $\opt+{1\over 2}\cE^{-}$ on this instance. Moreover, for any $n >3$ and any $\cE^+_{\infty}$ there exists a graph search instance on $n$ nodes with maximum positive error $\cE^+_{\infty}$ such that any randomized algorithm must incur expected cost at least $\opt+(n-2)\cE^+_{\infty}/2$ on this instance.
\end{proposition}

The full proof of Theorem~\ref{thm:searching-graphs-greedy} and the proof of Corollary~\ref{cor:decremental-error-search}, given in Section~\ref{sec:missing_proofs}, rely on a charging argument which shows that distance travelled away from the goal can be directly attributed to errors in the predictions of observed nodes. The proof of Theorem~\ref{thm:optimal_E1_lowerbounds} and \Cref{lem:lower-bounds-randomized} are constructive and can also be found in Section~\ref{sec:missing_proofs}.

For completeness we now sketch the proof of Theorem~\ref{thm:searching-graphs-greedy}. The proof considers the progress $\Delta_i$ as in \Cref{eq:defn-Delta},
which measures how much closer the agent is to the goal after the $i^{th}$ iteration of the algorithm. The cost of the algorithm is given by $\alg = \sum_{i\in[T]} d_{G_i}(v_{i-1},v_i)$. At each step $i\in [T]$, one can show that the distance travelled $d_{G_i}(v_{i-1},v_i)$ in the observed subgraph $G_i$ is bounded above by the sum of three terms:
\begin{equation}\label{eq:distance-travelled-decomposition}
    d_{G_i}(v_{i-1},v_i) \leq \Delta_i + \cE^-(v_i) + \cE^+(w_i),
\end{equation}
where $\cE^-(v_i)$ is the negative error at $v_i$, and $\cE^+(w_i)$ is the positive error at some vertex $w_i$ on a shortest path from $v_{i-1}$ to $g$. In particular, all three terms in this upper bound are independent of the observed subgraph $G_i$. The statement of the theorem then follows by summing both sides of \Cref{eq:distance-travelled-decomposition} over all $i\in[T]$.

\section{Exploration Under Relative Error}\label{sec:multiplicative_error}
In this section, we consider the setting when the prediction function satisfies Equation~\eqref{eq:multiplicative_error_setup} for every $v\in V$. We assume that $\varepsilon \in (0,1)$: in particular, if $\varepsilon \geq 1$, then $f(v) = 0$ is a valid prediction at every vertex and no exploration algorithm can avoid visiting the entire graph in the worst case.

If $\varepsilon$ is known to the searcher a priori then given access to the prediction at a node, the searcher can construct an upper bound on the true distance-to-goal. On trees this allows one to limit exploration to a ball of some radius $R$ (dependent on $\varepsilon$ and $\opt$) around the initial vertex, effectively ``pruning'' distant nodes from the vertex set. In particular, one could limit their search to the set:
\begin{equation}\label{eq:defn-S}
    S_{\varepsilon, r} \defeq \left\{v\in V \mid d_G(v,r) \leq \frac{1}{1-\varepsilon} f(r)\right\}.
\end{equation}
We couple this observation with the algorithm in the previous section and obtain the following algorithm.
\begin{algorithm}[H]
    \begin{algorithmic}
    \caption{ \epsknownalgo$(G,r, \varepsilon)$ }\label{alg:vareps-known-search}
        \State $v_0 \gets r$
        \State $i \gets 0$
        \State $V_0 \gets \{r\}$
        \While{$v_i \neq g$}
            \State $i\gets i+1$
            \State $v_i \in \argmin_{v\in \partial V_{i-1}\cap S_{\varepsilon, r}} d_{G_i}(v_{i-1},v) + f(v)$
            \State $V_i \gets \{v_0, ... , v_i\}$
        \EndWhile
    \end{algorithmic}
\end{algorithm}

In Section~\ref{sec:missing_proofs} we leverage favorable properties of the set $S_{\varepsilon, r}$ to show that Algorithm~\ref{alg:vareps-known-search} satisfies the guarantees of Theorem~\ref{cor:multiplicative-error-trees}. We observe that in particular, Theorem~\ref{cor:multiplicative-error-trees} implies that for any $n$ the competitive ratio incurred by Algorithm~\ref{alg:vareps-known-search} tends to 1 as $\varepsilon \to 0$. The combination of the truncation with the shortest-path rule in Algorithm~\ref{alg:vareps-known-search} was necessary to secure this property: for example, a simple scheme such as running breadth-first-search on the truncated set $S_{\varepsilon, r}$ would not enjoy such a guarantee in the worst case.

We note that Algorithm~\ref{alg:vareps-known-search} crucially relies on the fact that the searcher knows the value of $\varepsilon$ and so it cannot be deployed in the setting where $\varepsilon$ is unknown. For the latter regime we propose an alternative algorithm, also based on Algorithm~\ref{alg:greedy-l1-search}, which provably succeeds for unknown values of $\varepsilon$ under the assumption that $\varepsilon$ is small, e.g. $\varepsilon < 1/3$.
\begin{algorithm}[H]
    \begin{algorithmic}
    \caption{\epsuknownalgo$(G,r, \beta)$}\label{alg:weighted-vareps-unknown-search}
        \State $v_0 \gets r$
        \State $i \gets 0$
        \State $V_0 \gets \{r\}$
        \While{$v_i \neq g$}
            \State $i\gets i+1$
            \State $v_i \in \argmin_{v\in \partial V_{i-1}} \beta d_{G_i}(v_{i-1},v) + f(v)$
            \State  $V_i \gets \{v_0, ... , v_i\}$
        \EndWhile
    \end{algorithmic}
\end{algorithm}
In Section~\ref{sec:missing_proofs} we show that on trees this reweighting scheme ensures that Algorithm~\ref{alg:weighted-vareps-unknown-search} never explores nodes which are far from the goal and use this property to establish the guarantees in Theorem~\ref{thm:exploration-trees-epsilon-unknown} under the setting where $\beta = 2/3$. 

We give a lower bound to establish that Algorithm~\ref{alg:vareps-known-search} is almost optimal. Specifically, we show that even when $\varepsilon$ is known a-priori, the asymptotic dependence on $n\cdot \varepsilon$ is tight up to factors of $1/(1-\varepsilon)$: 
\begin{theorem}\label{theorem:multiplicative_error_explo_LB}
    For all $n$ sufficiently large ($n\geq 6$) and for any $\varepsilon \in (0,1)$, there exists an instance $\mathcal{I}$ of the exploration on weighted trees $G$ with predictions \eqref{eq:multiplicative_error_setup}, such that any algorithm for the exploration problem on $G$ must incur cost $\alg \geq (1+n\varepsilon)\opt$ on $\mathcal{I}$.
\end{theorem}
Note that this lower bound also applies to the regime addressed by Algorithm~\ref{alg:weighted-vareps-unknown-search}. Moreover, one can prove an analogous lower bound for the case of potentially randomized exploration strategies, as per the following result.

\begin{proposition}\label{prop:randomized-multiplicative-lower-bound}
    For all $n$ sufficiently large ($n\geq 6$) and for any $\varepsilon \in (0,1)$, there exists an instance $\mathcal{I}$ of the exploration on weighted trees $G$ with predictions \eqref{eq:multiplicative_error_setup}, such that any randomized algorithm for the exploration problem on $G$ must incur cost $\alg \geq (1+{n\varepsilon\over 2})\opt$ on $\mathcal{I}$.
\end{proposition}

\subsection{Planning With Relative Error}

Under the model of relative error described by Equation~\eqref{eq:multiplicative_error_setup}, planning problems become either trivial or impossible, with no intermediate regimes. In the context of predictions $f(v) = (1+\varepsilon_v) d_G (v,g)$ for some $\varepsilon_v \in [-\varepsilon,\varepsilon]$, we observe that it must be that $f(g) = 0$, independent of the value of $\varepsilon$. For the planning problem to be nontrivial, there must also occur vertices $v\neq g$ such that $f(v) = 0$. Thus, for nontrivial instances of the planning problem in this regime, $\varepsilon\geq 1$. However, instances with such (large) values of $\varepsilon$ are hopeless in the worst case: such a setting allows for the prediction of \textit{every} node to be set equal to 0, a case which forces the searcher to visit every node in the worst case.

\section{Planning}\label{sec:planning}

\subsection{Planning Bounds Via Metric Embeddings}

In this section, we analyze the performance of the algorithm by \citet{banerjee2022graph} for the planning problem, as a function of how similar the target graph is to some graph $G'$ admitting inexpensive tours. In the planning problem, the full graph $G$ as well as all predictions $f$ are made available to the algorithm upon initialization. \citet{banerjee2022graph} study the implied error functions $\phi_0:V\rightarrow \R$ and $\phi_1:V\rightarrow \R$, defined as
\[
    \phi_0(v) \defeq \left|\{u \in V: f(u)\neq d_{G}(u,v)\}\right|
\]
and
\[
    \phi_1(v) \defeq \sum_{u\in V} \abs{f(u)-d_{G}(u,v)}.
\]

\citet{banerjee2022graph} consider an algorithm that iteratively visits sublevel sets of $\phi_0$ or $\phi_1$ respectively, for geometrically increasing thresholds (see Algorithm~\ref{alg:banerjee-planning}). Their algorithm is very simple: for every threshold, it visits every node in the sublevel set before increasing to the next threshold value. Algorithmically, \citet{banerjee2022graph} accomplish this by computing a minimum length Steiner tree of the sublevel set, which is in general not computationally efficient. However, one can replace this computationally expensive procedure with a polynomial-time constant factor approximation (see e.g. \citep{karlin2021slightly}), which preserves the asymptotic upper bound on the algorithms' performance.

\begin{algorithm}
    \begin{algorithmic}
    \caption{\texttt{FullInfoX} $(G, r, \phi)$ from \citet{banerjee2022graph}}\label{alg:banerjee-planning}
        \State $V_i \gets r$
        \State $\lambda \gets 0$
        \While{$g \not \in V_i$}
            \State $V_i \gets V_i \cup  L^{-}_{\phi}(2^\lambda)$
            \State Compute a minimum length Steiner tree of $V_i$ and perform an Euler tour of the tree
            \State $\lambda \gets \lambda + 1$
        \EndWhile
    \end{algorithmic}
\end{algorithm}

This definition of Algorithm~\ref{alg:banerjee-planning} motivates our analysis, which focuses on the distortion of embedding the instance graph into some graph $G'$ which admits inexpensive tours. Recall the definition of $\tour_{G}(S)$ as in Equation~\eqref{eq:defn_tour}.
\begin{definition}[Easily-tourable]
    A graph $G' = (V', E')$ is $c$-\textit{easily-tourable} for some $c > 0$ if for any $S'\subseteq V'$,  $\tour_{G'} \leq c \cdot \text{diam}(S')$.
\end{definition}

In Section~\ref{sec:missing_proofs}, we establish that in easily-tourable graphs the algorithm of Banerjee \etal enjoys good performance. We then show that if a graph $G$ can be embedded into an easily-tourable graph $G'$ with distortion $\rho$, then $G$ itself must be easily-tourable with tour costs that scale with $\rho$. This culminates in the following result:
\begin{lemma}\label{lemma:planning-embedding}
    Given $G$ an unweighted graph, if $G$ admits an embedding $\tau:G\rightarrow G'$ of distortion $\rho$ for $G'$ some $c_{G'}$-easily-tourable graph, then Algorithm~\ref{alg:banerjee-planning} with objective $\phi=\phi_0$ from \citet{banerjee2022graph} incurs cost at most
    $
        \opt + O(\rho \cdot  c_{G'}\cdot \cE_0).
    $
    If $G$ has integer-valued distances and admits an embedding of distortion $\rho$ into $G'$ some easily-tourable graph, then Algorithm~\ref{alg:banerjee-planning} with objective $\phi=\phi_1$ incurs cost at most
    $
        \opt + O(\rho \cdot c_{G'} \cdot \cE_1).
    $
\end{lemma}

In particular, (weighted) paths and cycles are easily-tourable with respect to constant $c$, resulting in the guarantees in Theorem~\ref{theorem:planning-path-embedding-E0}. In Section~\ref{sec:path-embedding-doubling-dimension} we give results suggesting that our analysis of planning problems via metric embeddings is a refinement of the analysis of \citet{banerjee2022graph}.

\subsection{Lowerbounds For Planning}
In this section, we provide lowerbounds that extend the results in \citet{banerjee2022graph}. Consider the planning problem on some weighted graph $G$ with integer weights\footnote{We note that the analysis of Banerjee et al. for the above setting also appears to go through for the case non-integer weights, and that the lowerbound provided by Lemma~\ref{lemma:lowerbound-E1-planning} would then hold for that setting too.} where prediction error is parameterized by $\cE_1$. \citet{banerjee2022graph} propose an algorithm which provably incurs cost at most
$
    \opt + O(\cE_1 \texttt{poly}(\lambda)),
$
where $\lambda$ is the doubling constant of the input graph. We provide complementary lower bounds by establishing the following result:

\begin{lemma}\label{lemma:lowerbound-E1-planning}
    Let $\mathcal{A}$ be any algorithm for the planning problem which is guaranteed to incur cost:
    $
        \opt + O(\cE_1^{a}\lambda^{b})
    $
    for some $a, b\in \R$ when run on a weighted graph $G$ with doubling constant $\lambda$, and with a error vector $\vec{e}$ such that $\|e\|_1 = \cE_1$. Then it must be the case that $a\geq 1$. Moreover, if $a = 1$, then $b\geq 1$.
\end{lemma}

\begin{figure}
    \centering
    \input{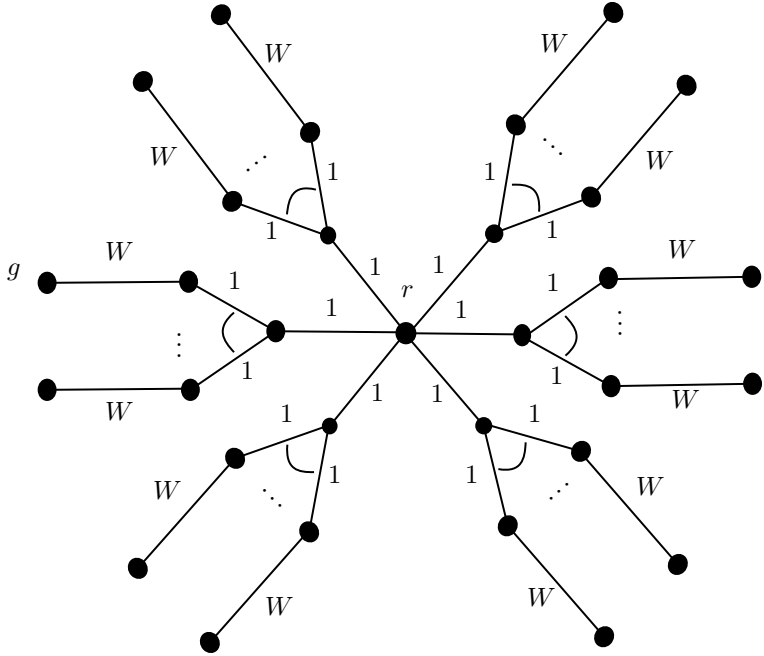}
    \caption{The lower bound construction for the proof of Lemmas~\ref{lemma:lowerbound-E1-planning} and \ref{lemma:lowerbound-E1-planning-trees}.}
    \label{fig:planning-E1-lowerbound}
\end{figure}

Our lowerbounds are constructive, and consider a family of graphs illustrated in Figure~\ref{fig:planning-E1-lowerbound}. A full proof can be found in Section~\ref{sec:missing_proofs}.

It is known that the doubling constant is always at least as large as the maximum degree but we note that in general it may be much larger even if one restricts themselves to trees. We analyze the algorithm of \citet{banerjee2022graph} and prove that in trees with integer weights performance of their algorithm can be bounded in terms of maximum degree, at the cost of paying a higher asymptotic dependence on $\cE_1$:
\begin{lemma}\label{lemma:planning-on-trees}
    Given $G$ a tree with integer-valued distances and maximum degree $\Delta$, consider the planning problem on $G$ with predictions satisfying $\cE_1~\geq~1$. Then on this problem instance Algorithm~\ref{alg:banerjee-planning} with objective $\phi = \phi_1$ from \citet{banerjee2022graph} incurs cost at most
    $
        \opt + O(\Delta \cE_1^2).
    $
\end{lemma}

We establish corresponding lowerbounds via a similar construction to that in Lemma~\ref{lemma:lowerbound-E1-planning}:
\begin{lemma}\label{lemma:lowerbound-E1-planning-trees}
    Let $\mathcal{A}$ be any algorithm for the planning problem  on trees with integer weights which is guaranteed to incur cost:
    $
        \opt + O(\cE_1^{a}\Delta^{b})
    $
    for some $a, b\in \R$ when run on a tree $G$ with maximum degree $\Delta$ and with a prediction vector $\vec{e}$ such that $\|e\|_1 = \cE_1$. Then it must be the case that $a\geq 1$ and $b\geq 1$. Moreover, it must be that $a + b \geq 3$.
\end{lemma}

\section{Numerical Experiments: Impact of Random Errors}\label{sec:experiments}

Throughout this paper we have focused on worst-case theoretical guarantees. In this section, we provide numerical results exploring the effectiveness of Algorithms~\ref{alg:greedy-l1-search} and ~\ref{alg:weighted-vareps-unknown-search} on exploration problems beyond the worst-case setting, particularly under random errors. The experiments show that in addition to being robust to adversarial error, the algorithms considered perform well in instances with stochastic error. Moreover, we find that although the guarantees for Algorithm~\ref{alg:weighted-vareps-unknown-search} were shown only for trees, empirically the algorithm also succeeds on general (cyclic) graphs. These results suggest that Algorithms~\ref{alg:greedy-l1-search} and ~\ref{alg:weighted-vareps-unknown-search} could be deployed effectively for exploration problems in practice.

We study the performance of Algorithm~\ref{alg:greedy-l1-search} in the absolute regime and Algorithm~\ref{alg:weighted-vareps-unknown-search} in the relative regime. In the left subfigure of Figure~\ref{fig:main_body_experiments_absolute} we plot the performance of Algorithm~\ref{alg:greedy-l1-search} for different graph topologies when the error is sampled at random in an absolute fashion. We plot the performance against the total $\ell_1$-norm of the error vector. %
In the right subfigure of Figure~\ref{fig:main_body_experiments_absolute} we explore the performance of Algorithm~\ref{alg:weighted-vareps-unknown-search} against relative error. Once again, we find that the algorithm enjoys empirical performance superior to that predicted by the worst-case upper bound. In particular, the gap between the worst-case bound and the average empirical performance is consistent over families of graphs with very different topologies.

In Table~\ref{table:experiment_table} we report the average ratio of algorithmic cost incurred to value of the upper-bound in Theorem~\ref{thm:searching-graphs-greedy} (given as a percentage). We compute this percentage for different classes of graphs over 100 runs of these experiments when the error, initial node and goal node, and graph structure (when applicable) have been sampled at random.

\begin{figure*}[h]
  \centering
  \begin{minipage}[c]{0.35\linewidth}
  \includegraphics[scale=0.35]{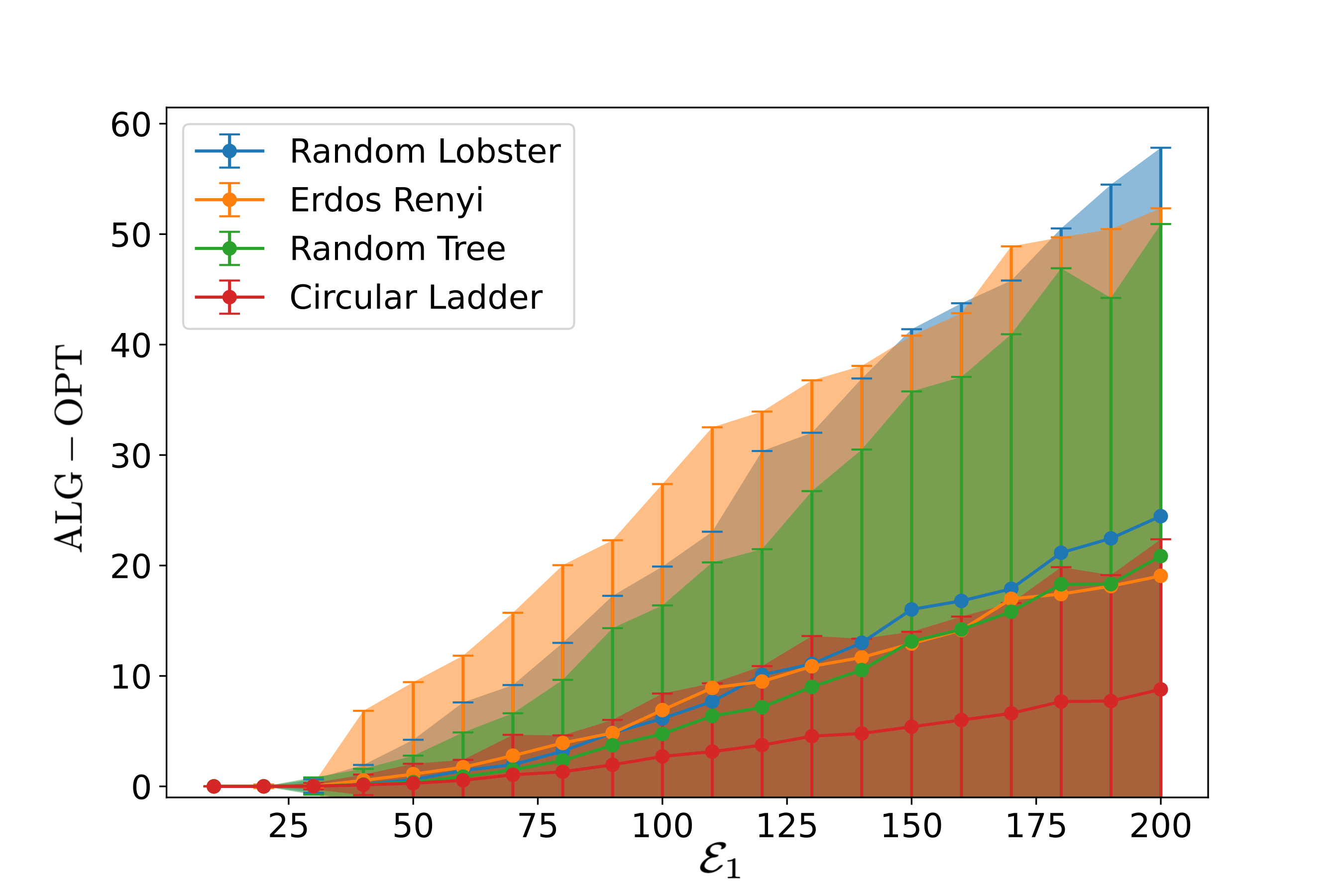}
  \end{minipage}
  \hfill
  \begin{minipage}[c]{0.5\linewidth}
  \includegraphics[scale=0.35]
  {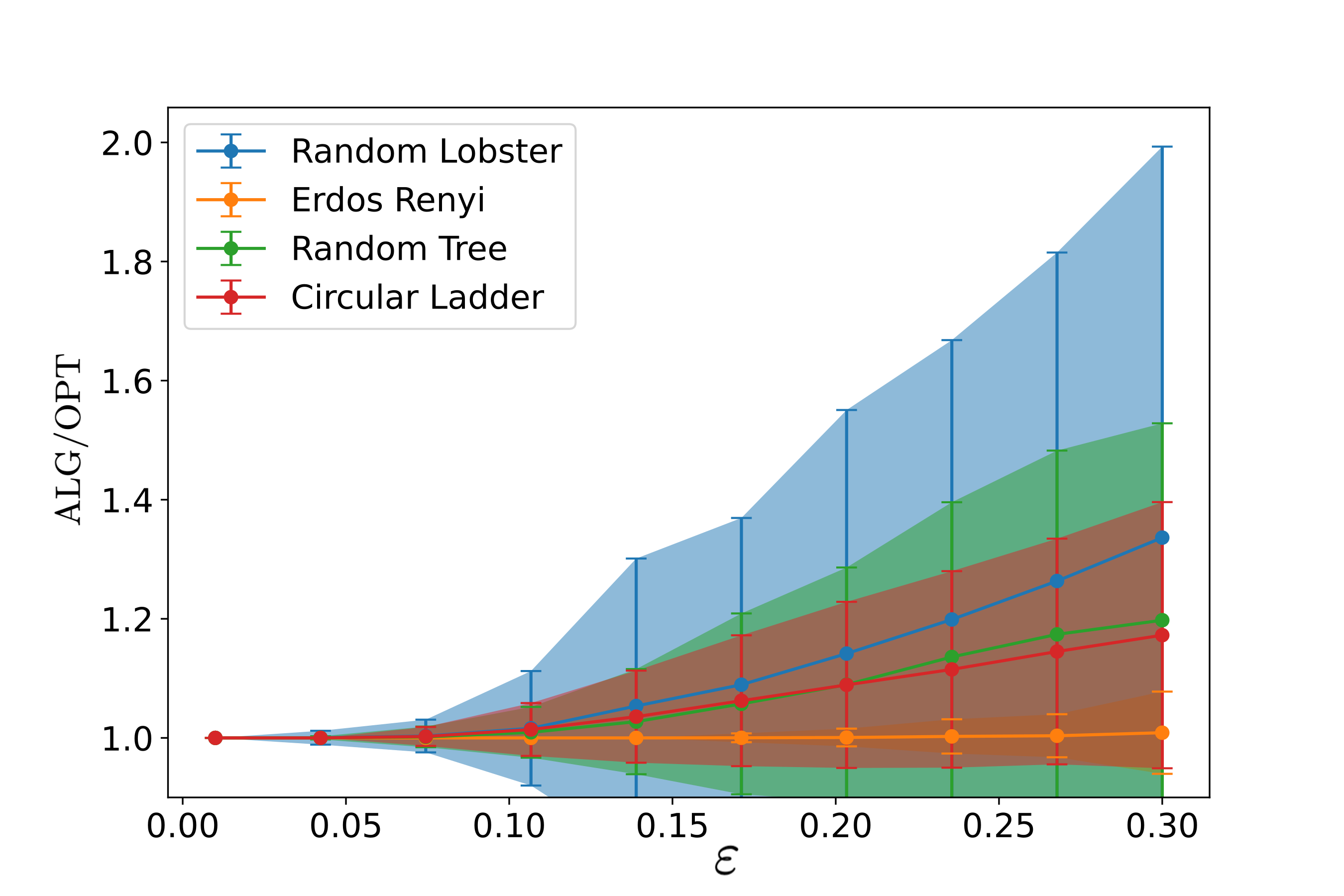}
  \end{minipage}
  \caption{Performance of Algorithms~\ref{alg:greedy-l1-search} and ~\ref{alg:weighted-vareps-unknown-search} against random errors. Experimental procedures are detailed in Section~\ref{sec:further_experiments}. The number of nodes is fixed over all graph topologies and error settings.  LEFT: Average and standard deviation of $\alg-\opt$ incurred by Algorithm~\ref{alg:greedy-l1-search} over 2000 independent random trials for varying values of $\cE_1$. RIGHT:  Average and standard deviation of $\alg/\opt$ incurred by Algorithm~\ref{alg:weighted-vareps-unknown-search} over 2000 independent random trials for varying values of $\varepsilon$.
  }

    \label{fig:main_body_experiments_absolute}
\end{figure*}

\begin{figure}[t]
  \centering
    \includegraphics[scale=0.5]{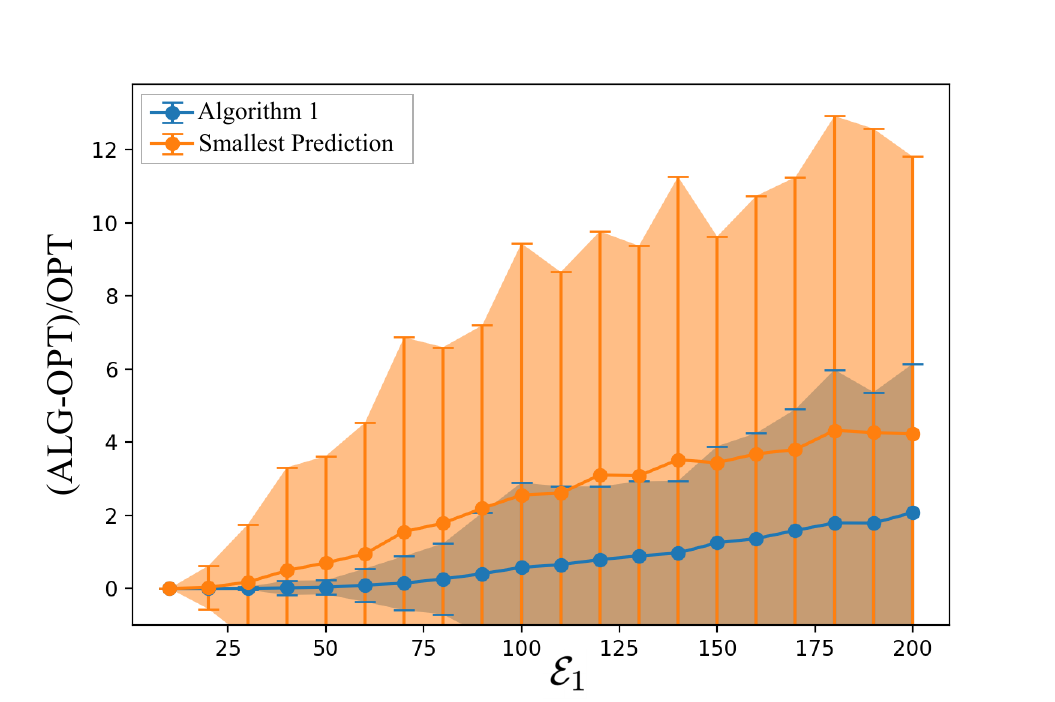}
  \caption{A comparison of the performance of \Cref{alg:greedy-l1-search} with the Smallest Prediction heuristic. We plot the average and the standard deviation of the performance of \Cref{alg:greedy-l1-search} and that of the Smallest Prediction heuristic against the magnitude of the error vector $\cE_1$. Experiments in this figure were conducted on random trees; for analogous results on other graph topologies, see Figure~\ref{fig:baseline_comparison_all} in \Cref{sec:further_experiments}.}\label{fig:comparison_to_baseline}
\end{figure}

We also empirically compare the performance of \Cref{alg:greedy-l1-search} to another natural heuristic, which we call Smallest Prediction. In Smallest Prediction, the agent always travels to the vertex $v$ in $\partial V_{i-1}$ with the smallest value of the prediction function $f(v)$. While our theoretical results already show that \Cref{alg:greedy-l1-search} achieves optimal performance in the worst-case, we show that our algorithm performs better than Smallest Prediction in the presence of random error across a variety of graph topologies. In \Cref{fig:comparison_to_baseline}, we plot the average performance of \Cref{alg:greedy-l1-search} against the performance of Smallest Prediction (measured as the distance travelled by the agent, minus $\opt$, as a fraction of of $\opt$) in random trees with $100$ vertices for a growing value of the magnitude $\cE_1$ of the error vector. More details on this comparison can be found in the supplementary material.

Further details of our experiments, including details on the error models and the graph families being used in the experiments can be found in Section~\ref{sec:further_experiments} in the supplementary material.

\begin{table}[h]
\centering
\def\arraystretch{1.2}%
\begin{tabular}{c||cccc}
 \begin{tabular}[x]{@{}c@{}} {\bf {GRAPH} }\\ {\bf {FAMILY}} \end{tabular} & \begin{tabular}[x]{@{}c@{}} {Random}\\ {Lobster}\end{tabular} & \begin{tabular}[x]{@{}c@{}} {Erd\"os}\\{R\'enyi}\end{tabular}   & \begin{tabular}[x]{@{}c@{}} {Random}\\ {Tree}\end{tabular} & \begin{tabular}[x]{@{}c@{}} {Circular}\\ {Ladder}\end{tabular} \\ \hline\\
{\bf {COST}}{($\%$)} & {$2.4\pm 2.5$} & {$3.1\pm 4.3$} & {$1.7\pm 2.3$} & {$0.7\pm 0.5$} \\
\end{tabular}
\caption{Average empirical cost of running Algorithm ~\ref{alg:greedy-l1-search} as a percentage of the upperbound in Theorem~\ref{thm:searching-graphs-greedy} for different family of graphs. Experiments performed on graphs with 300 nodes. These results demonstrate that when run with randomly-generated errors, the actual cost incurred by the algorithm is a very small fraction of the upperbound.}\label{table:experiment_table}
\end{table}

\section{Conclusions and Future Directions}

In this work we have introduced new general algorithms for the problem of searching in an unknown graph. Under the absolute error model we design algorithms which succeed in a broad class of graphs and prove that these algorithms are optimal (Section~\ref{sec:explo_L1_error}). We then move beyond the absolute error regime and consider relative error; to the best of the authors' knowledge, this work is the first to address the exploration problem under this natural error model. Within this setting we propose algorithms for the exploration problem on weighted trees and show that their performance is nearly-optimal (Section~\ref{sec:multiplicative_error}). 

We complement our advances in the exploration setting by expanding the landscape of results for the planning problem. We extend the work of \citet{banerjee2022graph} by providing alternative performance guarantees which establish a linear--rather than quadratic--dependency on the error parameter $\cE_0$ in some graph families, and which suggests that such a lower asymptotic dependency may be attainable in general (Theorem~\ref{theorem:planning-path-embedding-E0}). We also complete the results of Banerjee \etal in the planning setting on integer-distance graphs by proving one cannot improve the factor of $\cE_1$ in their upperbound and that achieving this linear dependence on the error requires cost linear in the doubling constant $\lambda$ of the instance graph (Lemma~\ref{lemma:lowerbound-E1-planning}).

The work in this paper directly suggests several avenues for further study. While our lowerbounds demonstrate the impossibility of uniformly improving the results in Theorem~\ref{thm:searching-graphs-greedy}, it is possible the bound may be overly pessimistic in certain classes of graphs; it would be interesting to consider whether making stronger structural assumptions about the instance graph would yield better guarantees. In the setting of relative error, an immediate open problem is whether the guarantees on Algorithms~\ref{alg:vareps-known-search} and \ref{alg:weighted-vareps-unknown-search} can be extended to more general graphs. In order to improve understanding of the planning problem, a complete characterization of easily-tourable graphs would  better contextualize Lemma~\ref{lemma:planning-embedding}. Finally, while the numerical results suggest the algorithms proposed in this paper perform well under random errors, formal guarantees studying this setting would be a valuable addition.

\bibliographystyle{abbrvnat}
\bibliography{sample}

\appendix

\section{Survey of Related Works}\label{ssec:related_work}

\paragraph{Graph search with distance-to-goal predictions} The problem and prediction settings considered in this work most closely correspond to those considered by \citet{banerjee2022graph}. \citet{banerjee2022graph} study exploration and planning under absolute error models. They consider two parametrizations of prediction error: the first in terms of the number of nodes at which predictions are not equal to true distance-to-goal, denoted $\cE_0$, and the second in terms of the $\ell_1$ norm of the vector of errors, denoted $\cE_1$ as in this work. They develop an algorithm for exploration on unweighted trees, and prove guarantees on its performance in terms of $\cE_0$. They also develop algorithms for planning on graphs: on unweighted graphs, they establish guarantees parameterized by $\cE_0$, while in graphs with integer-valued distances their performance bounds are parameterized by $\cE_1$. %

\paragraph{Treasure Hunt} In the \textit{treasure hunt problem}, a mobile agent must traverse some unknown environment, continuous or discrete, to locate a stationary hidden goal \citep{alpern2006theory}. When the search environment is a graph, this problem shares many features with the exploration problem considered in this paper. \citet{bouchard2021almost} study the graph treasure hunt problem when the searcher receives no additional information. They establish lower bounds on the total cost incurred by any algorithm in terms of the number of edges in the ball of radius $\opt$ around the root node, and give algorithms with performance guarantees which asymptotically exceed these lower bounds by at most a factor of $\log(\opt)$. Graph treasure hunt problems have also been considered when the agent receives help with the task. \citet{komm2015treasure} study the case when the the searcher receives generic advice, which can take the form of any bit string.
They consider the advice complexity of the treasure hunt task; they prove that there is an algorithm which achieves competitive ratio $r$ by receiving $O(n/r)$ bits of advice along the search and moreover they establish that any algorithm achieving a competitive ratio of $r$ must receive $\Omega(n/r)$ bits of advice (see Theorems 4 and 5 in \cite{komm2015treasure}). In the setting where graph vertices are anonymized, i.e. the searcher has no way to recognize whether a vertex has or has not previously been visited, recent work has studied the task of graph treasure hunt with access to advice from an omniscient oracle which marks vertices with binary labels \citep{bhattacharya2022treasure,gorain2022pebble}.
 
The exploration model studied in this paper can be viewed as a graph treasure hunt problem with specific kinds of advice (predictions of distance-to-goal). One main contrast with this work is that the advice considered can contain adversarial errors, and indeed the impact that different error models have on the graph exploration task is a core topic investigated in this work.

\paragraph{Path-Planning and A$^*$ Search} Distance-to-goal predictions have been the subject of study in many path-planning and graph-traversal settings. Initialized with a root node and some known target node, \textit{path-planning problems} seek to learn a shortest path between the root and goal and common problem models assume access to a set of distance-to-goal predictions, referred to as ``heuristics'' within this literature \cite{pohl1969bi, ferguson2005guide}. A$^*$ search is a celebrated algorithm for the problem of path-finding and graph-traversal, designed for cases when the entire graph $G$ and all predictions $f$ are accessible in memory, and has spawned many algorithmic variants  \citep{rios2010survey,foead2021systematic,paliwal2023survey}. 

\begin{figure}         
  \centering
  \includegraphics[scale=0.5]{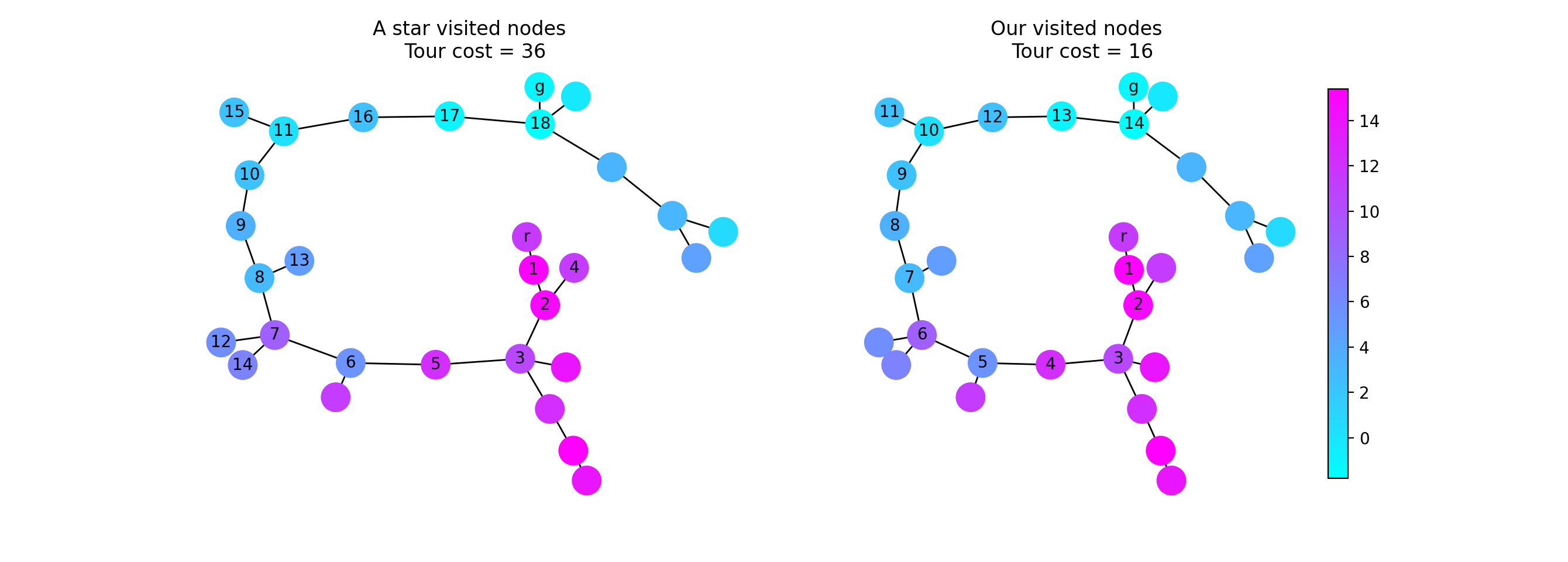}
  \caption{Comparing Algorithm~\ref{alg:greedy-l1-search} versus A$^*$ search on a random tree with randomly generated errors. The same set of predictions is provided to both algorithms. While the set of nodes visited by the two algorithms is comparable, the computational model in A$^*$ places no penalty on traversal distance so that algorithm has a tendency to double-back on itself, leading to a more expensive tour. Nodes are colored by prediction value and labeled with the order in which they are first visited by the relevant algorithm. Tour cost is taken to be $\sum d(v_i, v_{i+1})$ for indices $i$ ordered according to when a node is first visited: the tour cost for A$^*$ omits costs incurred by re-expanding a node within the execution of the algorithm.}\label{fig:comparing_with_astar}
\end{figure}

Much of the theory of A$^*$ search focuses on cases when predictions have particular structural properties: a prediction function $f:V\rightarrow\R$ is called {\em admissible} if the prediction at every node $v$ is never greater than the actual distance to the goal from $v$. A prediction function $f:V\rightarrow\R$ is {\em consistent} (or {\em monotone}) if for every node $v$ and every neighbor $u$ of $v$, $f(v) \le f(u) + d(v,u)$ and $f(g) =0$. Consistency is studied so heavily in part because it implies admissibility. Admissible heuristics are well-motivated and occur in various other problems. For example,~\citet{eden2022embeddings} consider access to an oracle that provides an underestimate for the probabilities of any element in some discrete probability distribution. In addition to being well-motivated by applications, admissibility is a focus of literature because A$^*$ search with admissible heuristics enjoys optimality properties \citep{dechter1985generalized}.

While many of the problems and algorithms within path-planning may appear closely related to this work at first-glance, we emphasize that the goals and cost models differentiate the graph searching problems considered in this work from those in path-planning. In particular, the design and algorithmic guarantees on path-planning algorithms like A$^*$ search implicitly assume that the full graph and predictions are available in memory upon initialization of the algorithm. Performance guarantees and notions of optimality are proven in terms of computational procedures, rather than traversal distance.  For example, when predictions are admissible A$^*$ is optimal in the sense that the set of nodes expanded (an operation analogous to visiting a node) is minimal \citep{dechter1985generalized}. However, the sequence in which these nodes are expanded can incur high traversal distance, as illustrated in Figure~\ref{fig:comparing_with_astar}. More generally, the goals of algorithms for path-planning differ substantially from that in the graph search problem: the path-finding problem seeks to learn and return a shortest path even if the process required to find such a path is expensive in the sense of graph traversal, whereas the aim of a graph searching problem is finding the goal node in an inexpensive manner and makes no demand that a shortest path from the root node to the goal be in the set of visited nodes upon termination of the algorithm. These differences in goals and algorithmic guarantees mean that in cases when the environment is unknown, i.e. when the graph and predictions are not available in memory but must instead by accessed by traversal, path-planning algorithms may incur much higher traversal cost than the graph search algorithms proposed in this paper.

\paragraph{Learning-Based Algorithms.} Recent years have seen a marked increase  in the integration of machine learning techniques to enhance traditional algorithmic challenges. Such algorithms have been developed for various topics including online algorithms~\cite{lykouris2018competitive,purohit2018improving,angelopoulos2020online,wei2020optimal,bamas2020primal,aamand2022matchingdegrees,antoniadis2023online,anand2020customizing,diakonikolas2021learning,gupta2022augmenting}, data structures~\cite{kraska2017case,mitz2018model,ferragina2020learned,lin2022learning}, and streaming models~\cite{HsuIKV19,indyk2019learning,aamand2019frequency,jiang2020learningaugmented,cohen2020composable,du2021puttingthelearning,EdenINRSW21,ChenEILNRSWWZ22,li2023learning}. For an extensive collection of learning-based algorithms, refer to the repository at \href{https://algorithms-with-predictions.github.io/}{https://algorithms-with-predictions.github.io/}.

\section{Missing Proofs}\label{sec:missing_proofs}
In this section, we present detailed proofs of the results in the main body of the paper. 

\subsection{Proofs For Section~\ref{sec:explo_L1_error}}

\begin{proof}[{\bf Proof of Theorem~\ref{thm:searching-graphs-greedy} and Corollary~\ref{cor:decremental-error-search}}]
Algorithm~\ref{alg:greedy-l1-search} follows a shortest path in $G_i$ from $v_{i-1}$ to $v_{i}$. A simple consequence of this is that:
\[
    \operatorname{ALG} = \sum_{i\in [T]}d_{G_i}(v_{i-1},v_i).
\]
Let $\Delta_i \defeq d_G(v_{i-1},g) - d_G(v_{i},g)$, and let $T$ be the number of iterations of the while loop executed, so that $v_T = g$. We then have:
\[
    \sum_{i\in [T]}\Delta_i = d_G(v_0, g) - d_G(v_T,g) =  d_G(v_0, g) - d_G(g,g) = \operatorname{OPT}.
\]
Consider any iteration $i \in [T]$. Let $w_i$ be the first vertex outside of $V_{i-1}$ encountered when traversing a shortest path from $v_{i-1}$ to $g$. Note that $w_i\in \partial V_{i-1}$, and hence, by the update rule in Algorithm~\ref{alg:greedy-l1-search} we have:
\begin{equation}\label{eq:bound-from-update-rule}
    f(v_i) + d_{G_i}(v_{i-1},v_i) \leq f(w_i) + d_{G_i} (v_{i-1}, w_i).
\end{equation}
Furthermore, by the definition of $w_i$, we have:
\begin{equation}\label{eq:true-distance-is-known-distance}
    d_{G_i}(v_{i-1}, w_i) = d_G(v_{i-1},w_i) .
\end{equation}
The above follows from a simple contradiction argument, for the existence of a shorter path from $v_{i-1}$ to $w_i$ in $G$ which is not in $G_i$, would contradict the definition of $w_i$. We also have:
\begin{equation}\label{eq:distance-is-sum-of-distances}
    \hspace{0.5cm} d_G(v_{i-i},g) = d_G(v_{i-1},w_i) + d_G(w_{i},g).
\end{equation}
We then have, for any $i \in [T]$:
\begin{align*}
    d_{G}(v_i,g)-f(v_i) &\overset{\eqref{eq:bound-from-update-rule}}{\geq} d_G(v_i,g) +d_{G_i}(v_{i-1},v_i)-d_{G_i}(v_{i-1},w_i)-f(w_i)\\
    &\overset{\eqref{eq:true-distance-is-known-distance}}{=} d_G(v_i,g) +d_{G_i}(v_{i-1},v_i)-d_{G}(v_{i-1},w_i)-f(w_i)\\
    &= d_G(v_{i-1},g) - \Delta_i +d_{G_i}(v_{i-1},v_i)-d_{G}(v_{i-1},w_i)-f(w_i)\\
    &\overset{\eqref{eq:distance-is-sum-of-distances}}{=} d_{G_i}(v_{i-1},v_i) - \Delta_i + d_G(w_i,g) -f(w_i).
\end{align*}

\noindent 
Note that the vertices in the sequence $\{v_i\}_{i\in[T]\cup\{0\}}$ are always distinct, while the vertices in the sequence $\{w_i\}_{i\in[T]}$ might not be. This also implies that $T \leq n$. The above then implies:
\begin{align*}
     \operatorname{ALG}&=\sum_{i\in [T]}d_{G_i} (v_{i-1},v_i)\\
     &\leq \sum_{i\in [T]} \Delta_i  + d_{G}(v_i,g)-f(v_i)+ f(w_i)  - d_G(w_i,g)\\
     &= \sum_{i\in [T]} \Delta_i  + \sum_{i\in [T]} d_{G}(v_i,g)-f(v_i)+ \sum_{i\in [T]}f(w_i)  - d_G(w_i,g)\\
     &\leq \operatorname{OPT} + \,\cE_1^- + T\cdot \cE_\infty^+\\
     &\leq \operatorname{OPT} + \,\cE_1^- + n\cdot \cE_\infty^+.
\end{align*}
This completes the proof of Theorem~\ref{thm:searching-graphs-greedy}. When predictions are admissible, $\cE_1^- = \cE_1$ and $\cE_\infty^+=0$  so Corollary~\ref{cor:decremental-error-search} follows.
\end{proof}

\begin{proof}[{\bf Proof of Theorem~\ref{thm:optimal_E1_lowerbounds}}] We begin by proving the first half of the theorem. Given $\cE_1^{-}>0$ one considers the three-vertex graph path $P_3$ where the two edges are weighted with weight $w = \cE^{-}/2$, the start vertex / root is chosen to be the middle vertex and the goal is one of the other two vertices (see left side of Figure~\ref{fig:Lowerbound1}). Note that the value of $\opt$ is $d_G(r,g) = w$. When the predictions on the vertices are given by: $f(v_1)=0$ , $f(v_2)=w$ and $f(v_3)=0$, the error is equal to $\cE^{-}$ and the graph looks completely symmetric to the searcher, and hence in the worst case to find the goal, the searcher has to incur a cost of $3w= w + 2w = \opt + \cE^{-}$ as needed.

For the second part of the theorem, we construct a star on $n$ vertices, where each edge has weight $w=\cE_\infty^+/2$, the starting vertex is at the center of the star, and the goal is chosen arbitrarily among the other vertices. The prediction at the goal is then picked to equal $\cE_\infty^+$ so that it equals the prediction in all other vertices, i.e. we set the predictions to $f(r) = w$ and $f(v)=2w$ for all $w\neq r$. Every algorithm will then have to visit the entire star in the worst case, incurring a cost of $\cE^+_{\infty}/2 + \cE_{\infty}^+ (n-2) = \opt + \cE_{\infty}^+ (n-2)$ as needed (See the right side of Figure~\ref{fig:Lowerbound1}).
\end{proof}

\begin{proof}[{\bf Proof of \Cref{lem:lower-bounds-randomized}}]
    We apply Yao's minimax principle~\citep{yao1977probabilistic} to the same constructions used in the proof of Theorem~\ref{thm:optimal_E1_lowerbounds}. For both constructions, we consider the distribution over instances produced by choosing the goal node uniformly at random among the leaf nodes.
    
    In particular, for the first statement, we consider the performance of any deterministic algorithm on the distribution of instances given by taking the three-node path graph on the left-hand side of \Cref{fig:Lowerbound1}, and placing the goal $g$ at either $v_1$ or $v_3$ with equal probability. We fix predictions $f(v_1) = f(v_3) = 0$ and $f(v_2) = w$ as in the proof of \Cref{thm:optimal_E1_lowerbounds}. The expected cost incurred by any deterministic algorithm over this distribution of instances is $2w = \opt + w = \opt + \nicefrac{1}{2}\cE^{-}$. Yao's minimax principle then implies the stated lower bound for all randomized algorithms.
    
    The proof of the second result follows analogously by considering the second construction in the proof of \Cref{thm:optimal_E1_lowerbounds}.
\end{proof}

    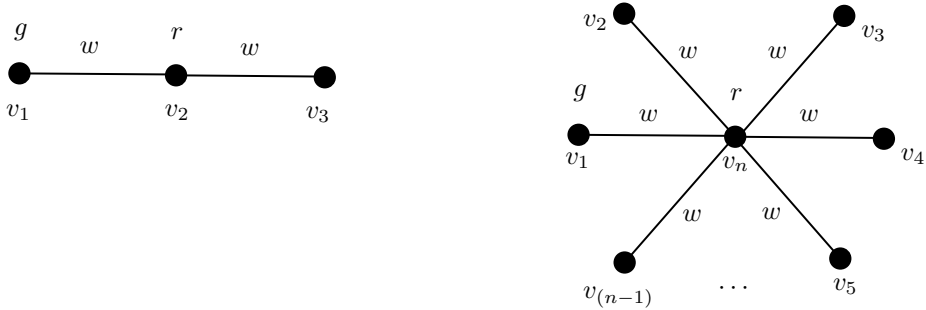
\begin{figure}
        \centering
        \tikzset{every picture/.style={line width=0.75pt}} %

\begin{tikzpicture}[x=0.75pt,y=0.75pt,yscale=-1,xscale=1]
\draw  [fill={rgb, 255:red, 0; green, 0; blue, 0 }  ,fill opacity=1 ] (57,71.12) .. controls (57,68.29) and (59.29,66) .. (62.12,66) .. controls (64.94,66) and (67.23,68.29) .. (67.23,71.12) .. controls (67.23,73.94) and (64.94,76.23) .. (62.12,76.23) .. controls (59.29,76.23) and (57,73.94) .. (57,71.12) -- cycle ;
\draw    (59.12,71.12) -- (143.23,71.7) ;
\draw  [fill={rgb, 255:red, 0; green, 0; blue, 0 }  ,fill opacity=1 ] (136,72.12) .. controls (136,69.29) and (138.29,67) .. (141.12,67) .. controls (143.94,67) and (146.23,69.29) .. (146.23,72.12) .. controls (146.23,74.94) and (143.94,77.23) .. (141.12,77.23) .. controls (138.29,77.23) and (136,74.94) .. (136,72.12) -- cycle ;
\draw    (141.12,72.12) -- (211.23,72.7) ;
\draw  [fill={rgb, 255:red, 0; green, 0; blue, 0 }  ,fill opacity=1 ] (211,73.12) .. controls (211,70.29) and (213.29,68) .. (216.12,68) .. controls (218.94,68) and (221.23,70.29) .. (221.23,73.12) .. controls (221.23,75.94) and (218.94,78.23) .. (216.12,78.23) .. controls (213.29,78.23) and (211,75.94) .. (211,73.12) -- cycle ;
\draw  [fill={rgb, 255:red, 0; green, 0; blue, 0 }  ,fill opacity=1 ] (339,102.12) .. controls (339,99.29) and (341.29,97) .. (344.12,97) .. controls (346.94,97) and (349.23,99.29) .. (349.23,102.12) .. controls (349.23,104.94) and (346.94,107.23) .. (344.12,107.23) .. controls (341.29,107.23) and (339,104.94) .. (339,102.12) -- cycle ;
\draw    (341.12,102.12) -- (425.23,102.7) ;
\draw  [fill={rgb, 255:red, 0; green, 0; blue, 0 }  ,fill opacity=1 ] (418,103.12) .. controls (418,100.29) and (420.29,98) .. (423.12,98) .. controls (425.94,98) and (428.23,100.29) .. (428.23,103.12) .. controls (428.23,105.94) and (425.94,108.23) .. (423.12,108.23) .. controls (420.29,108.23) and (418,105.94) .. (418,103.12) -- cycle ;
\draw    (423.12,103.12) -- (493.23,103.7) ;
\draw  [fill={rgb, 255:red, 0; green, 0; blue, 0 }  ,fill opacity=1 ] (493,104.12) .. controls (493,101.29) and (495.29,99) .. (498.12,99) .. controls (500.94,99) and (503.23,101.29) .. (503.23,104.12) .. controls (503.23,106.94) and (500.94,109.23) .. (498.12,109.23) .. controls (495.29,109.23) and (493,106.94) .. (493,104.12) -- cycle ;
\draw  [fill={rgb, 255:red, 0; green, 0; blue, 0 }  ,fill opacity=1 ] (473,42.12) .. controls (473,39.29) and (475.29,37) .. (478.12,37) .. controls (480.94,37) and (483.23,39.29) .. (483.23,42.12) .. controls (483.23,44.94) and (480.94,47.23) .. (478.12,47.23) .. controls (475.29,47.23) and (473,44.94) .. (473,42.12) -- cycle ;
\draw  [fill={rgb, 255:red, 0; green, 0; blue, 0 }  ,fill opacity=1 ] (362,41) .. controls (362,38.17) and (364.29,35.88) .. (367.12,35.88) .. controls (369.94,35.88) and (372.23,38.17) .. (372.23,41) .. controls (372.23,43.83) and (369.94,46.12) .. (367.12,46.12) .. controls (364.29,46.12) and (362,43.83) .. (362,41) -- cycle ;
\draw    (367.12,41) -- (423.12,103.12) ;
\draw    (478.12,42.12) -- (425.23,102.7) ;
\draw  [fill={rgb, 255:red, 0; green, 0; blue, 0 }  ,fill opacity=1 ] (471,164.58) .. controls (471,167.41) and (473.29,169.7) .. (476.12,169.7) .. controls (478.94,169.7) and (481.23,167.41) .. (481.23,164.58) .. controls (481.23,161.76) and (478.94,159.47) .. (476.12,159.47) .. controls (473.29,159.47) and (471,161.76) .. (471,164.58) -- cycle ;
\draw    (476.12,164.58) -- (423.23,104) ;
\draw  [fill={rgb, 255:red, 0; green, 0; blue, 0 }  ,fill opacity=1 ] (372.47,166.58) .. controls (372.47,169.41) and (370.18,171.7) .. (367.35,171.7) .. controls (364.52,171.7) and (362.23,169.41) .. (362.23,166.58) .. controls (362.23,163.76) and (364.52,161.47) .. (367.35,161.47) .. controls (370.18,161.47) and (372.47,163.76) .. (372.47,166.58) -- cycle ;
\draw    (367.35,166.58) -- (420.23,106) ;

\draw (91,53.4) node [anchor=north west][inner sep=0.75pt]    {$w$};
\draw (172,54.4) node [anchor=north west][inner sep=0.75pt]    {$w$};
\draw (54,85.4) node [anchor=north west][inner sep=0.75pt]    {$v_{1}$};
\draw (134,85.4) node [anchor=north west][inner sep=0.75pt]    {$v_{2}$};
\draw (205,86.4) node [anchor=north west][inner sep=0.75pt]    {$v_{3}$};
\draw (137,46.4) node [anchor=north west][inner sep=0.75pt]    {$r$};
\draw (58,44.4) node [anchor=north west][inner sep=0.75pt]    {$g$};
\draw (373,88) node [anchor=north west][inner sep=0.75pt]    {$w$};
\draw (454,88) node [anchor=north west][inner sep=0.75pt]    {$w$};
\draw (336,109.4) node [anchor=north west][inner sep=0.75pt]    {$v_{1}$};
\draw (345,40) node [anchor=north west][inner sep=0.75pt]    {$v_{2}$};
\draw (485.23,45.52) node [anchor=north west][inner sep=0.75pt]    {$v_{3}$};
\draw (419,77.4) node [anchor=north west][inner sep=0.75pt]    {$r$};
\draw (340,75.4) node [anchor=north west][inner sep=0.75pt]    {$g$};
\draw (438,56) node [anchor=north west][inner sep=0.75pt]    {$w$};
\draw (393,56) node [anchor=north west][inner sep=0.75pt]    {$w$};
\draw (505.23,107.52) node [anchor=north west][inner sep=0.75pt]    {$v_{4}$};
\draw (435,137.4) node [anchor=north west][inner sep=0.75pt]    {$w$};
\draw (471,173.4) node [anchor=north west][inner sep=0.75pt]    {$v_{5}$};
\draw (395,138.4) node [anchor=north west][inner sep=0.75pt]    {$w$};
\draw (345,177.4) node [anchor=north west][inner sep=0.75pt]    {$v_{( n-1)}$};
\draw (415,111.4) node [anchor=north west][inner sep=0.75pt]    {$v_{n}$};
\draw (413,175) node [anchor=north west][inner sep=0.75pt]    {$\cdots$};

\end{tikzpicture}
        \caption{The construction in the proof of Theorem~\ref{thm:optimal_E1_lowerbounds}.}
        \label{fig:Lowerbound1}
    \end{figure}

\subsection{Proofs For Section~\ref{sec:multiplicative_error}} 

\subsubsection{Algorithmic Guarantees Under Relative Error}

We begin the analysis of Algorithm~\ref{alg:vareps-known-search} by establishing the following properties of the set $S_{\varepsilon, r}$ defined in Equation~\eqref{eq:defn-S}. For $G$ a tree, let $P_G(u,v)$ denote the (unique) shortest path between nodes $u$ and $v$ in $G$.
\begin{lemma}\label{lem:basic-properties-of-set}
    For $S_{\varepsilon,r}$ as defined in \eqref{eq:defn-S} and $G$ a weighted tree, then the following hold:
    \begin{enumerate}
        \item[(i)] $\forall v\not\in S_{\varepsilon,r}$, $d_{G}(v,r) > OPT$,
        \item[(ii)] $\forall v\in S_{\varepsilon,r}$, $d_G(v,r) \leq \frac{1+\varepsilon}{1-\varepsilon} \cdot OPT$,
        \item[(iii)] $\forall v\in S_{\varepsilon, r}$, $d_{G}(v,g) \leq \frac{2}{1-\varepsilon}\cdot OPT$,
        \item[(iv)] $P_G(r,g) \subseteq S_{\varepsilon, r}$,
        \item[(v)] For any $v\in S_{\varepsilon, r}$, $P_G(v,g) \subseteq S_{\varepsilon,r}$.
    \end{enumerate}
\end{lemma}
\begin{proof}
    The properties follow immediately from the definition of the relative error model in Equation~\eqref{eq:multiplicative_error_setup} and the definition of $S_{\varepsilon,r}$ in Equation~\eqref{eq:defn-S}.
    \begin{itemize}
        \item[(i)] $\forall v\not\in S_{\varepsilon,r}$, $f(v) > \frac{1}{1-\varepsilon}f(r)$ and by Equation~\eqref{eq:multiplicative_error_setup}, $f(r) \geq (1-\varepsilon)d_G(r,g) = (1-\varepsilon)\opt$.
        \item[(ii)] $\forall v\in S_{\varepsilon,r}$, $d_G(v,r) \leq \frac{1}{1-\varepsilon} f(r)$, and by Equation~\eqref{eq:multiplicative_error_setup}, $f(r) \leq (1+\varepsilon)\opt$. 
        \item[(iii)] By triangle inequality, for any $v\in G$
        \[
            d_G(v,g) \leq d_G(v, r) + d_G(r,g) = d_G(v,r) + \opt
        \]
        and so by property (ii) above, for any $v\in S_{\varepsilon, r}$ the result follows.
        \item[(iv)] $\forall v\in P_G(r,g)$, $d_G(r,v) \leq d_G(r,g) = \opt$. Thus property (i) above implies $\forall v\in P_G(r,g)$, $v\in S_{\varepsilon,r}$.
        \item[(v)] Consider $v\in S_{\varepsilon,r}$, and let $u \defeq \argmin\{d_G(u,g) \mid u\in P_G(v,r)\}$. Because $G$ is a tree,
        \[
            P_G(v,g) = P_G(v,u) \cup P_{G}(u,g).
        \]
        In particular, $P_G(v,u) \subseteq P_G(v,r)$; observe that because $v\in S_{\varepsilon, r}$, $\forall w\in P_G(v,r)$, $d_{G}(w,r) \leq d_{G}(v,r) \leq \frac{1}{1-\varepsilon}f(r)$ and thus by the definition of $S_{\varepsilon,r}$, $P_G(v,r) \subseteq S_{\varepsilon,r}$. We've thus concluded that $P_G(v,u) \subseteq S_{\varepsilon,r}$.
        
        For the second portion of the path, $P_G(u,g)$, observe that by the definition of $u = \argmin\{d_G(u,g) \mid u\in P_G(v,r)\}$, it must be that $u\in P_G(r,g)$. Thus $P_G(u,g) \subseteq P_G(r,g)$ and so by property (iv), $P_G(u,g) \subseteq S_{\varepsilon,r}$.
    \end{itemize}
\end{proof}

With these properties, we now bound the distance travelled on the $i$th step of the algorithm: 
\begin{lemma}\label{lemma:multiplicative-error-ith-step}
    Let $G$ be a weighted tree, with predictions satisfying Equation~\eqref{eq:multiplicative_error_setup} with respect to parameter $\varepsilon < 1$. Then, the distance travelled by Algorithm~\ref{alg:vareps-known-search} on the $i$th iteration satisfies
    \begin{equation}\label{eq:multiplicative-error-ith-step-general}
        (1-\varepsilon)d_{G_i}(v_{i-1},v_i) \leq \Delta_i + 2\varepsilon d_{G}(v_{i},g).
    \end{equation}
    Additionally, if the predictions on $G$ are multiplicative and decremental,
    \begin{equation}\label{eq:multiplicative-error-ith-step-decremental}
        d_{G_i}(v_{i-1},v_i) \leq \Delta_i + \varepsilon d_{G}(v_i,g).
    \end{equation}
\end{lemma}
\begin{proof}[{\bf Proof of Lemma~\ref{lemma:multiplicative-error-ith-step}}]
    We observe that, by definition of Algorithm~\ref{alg:vareps-known-search}, for all iterations $i$, $V_{i-1} \subset S_{\varepsilon, r}$. Consider $w_i$ the first vertex outside of $V_{i-1}$ encountered when traversing $P(v_{i-1},g)$, the shortest path from $v_{i-1}$ to $g$. Note that $w_i \in \partial V_{i-1}$, and that by property (v) of Lemma~\ref{lem:basic-properties-of-set}, $w_i \in S_{\epsilon,r}$.

    Thus, on iteration $i$, $\exists w_i \in \partial V_{i-1} \cap S_{\varepsilon, r} \cap P(v_{i-1},g)$, so by the definition of Algorithm~\ref{alg:vareps-known-search} such $w_i$ satisfies
    \[
        f(v_i) + d_{G_i} (v_{i-1}, v_i) \leq f(w_i) + d_{G_i} (v_{i-1}, w_i).
    \]
    In particular, because $w_i \in P(v_{i-1}, g)\cap \partial V_{i-1}$ and by the tree properties of $G$, we have 
    \[
        d_{G_i}(v_{i-1},w_i) = d_{G}(v_{i-1},w_i) = d_{G}(v_{i-1},g) - d_{G}(w_i,g).
    \]
    We can thus upper bound
    \begin{align*}
        d_{G_i} (v_{i-1}, v_i) &\leq  d_{G}(v_{i-1},g) - d_{G}(w_i,g) + f(w_i) - f(v_i)\\
        &= d_{G}(v_{i-1},g) -d_{G}(v_i, g) + \big(d_{G}(v_i, g) - f(v_i)\big) + \big(f(w_i) - d_{G}(w_i,g)\big)\\
        &=\Delta_i -\varepsilon_{v_i} d_{G}(v_i, g) + \varepsilon_{w_i} d_{G}(w_i,g)
    \end{align*}
    where $\varepsilon_{v_i}, \varepsilon_{w_i} \in [-\varepsilon, \varepsilon]$ are the constants whose existence is implied by Equation~\eqref{eq:multiplicative_error_setup} such that
    \begin{equation}\label{eq:define_epsilon_v}
        f(v) = (1+\varepsilon_v) d_{G} (v,g).
    \end{equation}

    Consider two cases: first, consider the case when $\varepsilon_{w_i} > 0$. Because $w_i \in P(v_{i-1},g)$,
    \begin{equation*}
        d_G(w_i,g) \leq d_G(v_{i-1},g) \leq  d_{G}(v_{i-1},v_i) + d_{G}(v_i, g).
    \end{equation*}
    Combining this bound and the fact that $\varepsilon_{v_i}, \varepsilon_{w_i} \in [-\varepsilon, \varepsilon]$ yields:
    \[
        -\varepsilon_{v_i} d_{G}(v_i, g) + \varepsilon_{w_i} d_{G}(w_i,g) \leq 2\varepsilon d_{G}(v_i, g) + \varepsilon d_{G}(v_{i-1},v_i).
    \]
    Thus in this case, the desired bound in \eqref{eq:multiplicative-error-ith-step-general} is satisfied.

    In the second case, when $\varepsilon_{w_i} \leq 0$,
    \[
        -\varepsilon_{v_i} d_{G}(v_i, g) + \varepsilon_{w_i} d_{G}(w_i,g) \leq -\varepsilon_{v_i} d_{G}(v_i, g) \leq \varepsilon d_{G}(v_i, g) 
    \]
    which is trivially upper bounded by $2\varepsilon d_{G}(v_i, g) + \varepsilon d_{G}(v_{i-1},v_i)$. Thus, in both cases, the desired bound in \eqref{eq:multiplicative-error-ith-step-general} is satisfied.

    In the case of decremental errors, $\varepsilon_{v} \leq 0 \ \forall v\in V$. Thus the latter case always applies, so we can bound
    \[
        d_{G_i} (v_{i-1}, v_i) \leq \Delta_i -\varepsilon_{v_i} d_{G}(v_i, g) + \varepsilon_{w_i} d_{G}(w_i,g) \leq \Delta_i + \varepsilon d_{G}(v_i, g)
    \]
    thus establishing the bound in \eqref{eq:multiplicative-error-ith-step-decremental}.
\end{proof}

We are now equipped to prove Theorem~\ref{cor:multiplicative-error-trees}.
\begin{proof}[{\bf Proof of Theorem~\ref{cor:multiplicative-error-trees}}]
    We will show that Algorithm~\ref{alg:vareps-known-search} achieves the desired competitive ratio. We begin by establishing that the algorithm will terminate at the goal node $g$: by Property ($iv$) in Lemma~\ref{lem:basic-properties-of-set}, $P(r,g) \subseteq S_{\varepsilon, r}$, and further by definition $P(r,g)$ is connected, so Algorithm~\ref{alg:vareps-known-search} initialized at $r$ will explore a connected subgraph of $G$ that contains $g$, and will thus terminate at $g$.

    We now bound the total distance travelled by Algorithm~\ref{alg:vareps-known-search}. Consider the general case, when errors can be incremental or decremental. Then, by Lemma~\ref{lemma:multiplicative-error-ith-step},
    \begin{align*}
        \alg &= \sum_{i\in[T]} d_{G_i}(v_{i-1},v_i)  = \sum_{i\in[T]} (1-\varepsilon)d_{G_i}(v_{i-1},v_i)+\varepsilon d_{G_i}(v_{i-1},v_i)\\
        &\leq \sum_{i\in[T]} \Delta_i + \sum_{i\in[T]} 2\varepsilon d_{G}(v_i, g) + \sum_{i\in[T]} \varepsilon d_{G}(v_{i-1},v_i) = \opt + 2\varepsilon \left(\sum_{i\in [T]} d_{G}(v_i,g)\right) + \varepsilon \alg. 
    \end{align*}
    Because $V_i \subseteq S_{\varepsilon, r}$ for all iterations $i$, and by property ($iii$) in Lemma~\ref{lem:basic-properties-of-set}, 
    \[
        2\varepsilon \sum_{i\in [T]} d_{G}(v_i,g) \leq 2\varepsilon |S_{\varepsilon, r}| \cdot\frac{2}{1-\varepsilon} \opt.
    \]
    Thus, re-arranging,
    \[
        (1-\varepsilon) \alg \leq \opt\left(1 + |S_{\varepsilon, r}|\varepsilon \cdot\frac{4}{1-\varepsilon}\right),
    \]
    yielding the claimed competitive ratio from the trivial upper bound $|S_{\varepsilon,r}| \leq n$.

    In the case of decremental errors, Lemma~\ref{lemma:multiplicative-error-ith-step} and a similar argument give
    \[
        \alg = \sum_{i\in[T]} d_{G_i}(v_{i-1},v_i) \leq \sum_{i\in[T]} \Delta_i + \sum_{i\in[T]} \varepsilon d_{G}(v_i, g)  \leq \opt \left(1+ |S_{\varepsilon, r}|\varepsilon \cdot\frac{2}{1-\varepsilon}\right),
    \]
    yielding the claimed competitive ratio.
\end{proof}

To prove Theorem~\ref{thm:exploration-trees-epsilon-unknown}, we'll use the following lemmas: the first (Lemma~\ref{lemma:beta-weighted-update-bounded-dist-to-g}) establishes that certain algorithms never explore nodes too far from $g$. We emphasize that the below lemma makes use of distances in the full $G$, not $G_i$. In the case of weighted trees, these two distances are always identical: on a tree, for all iterations $i$ and for any $u, v \in V_i \cup \partial V_i$, $d_{G_i}(u,v) = d_{G}(u,v)$. The second lemma (Lemma~\ref{lemma:beta-weighted-alg-guarantee}) bounds the distance travelled by these algorithms on any given iteration.
\begin{lemma}\label{lemma:beta-weighted-update-bounded-dist-to-g}
    Consider the exploration problem on $G$ a weighted, undirected graph with predictions $f$ satisfying Equation~\eqref{eq:multiplicative_error_setup}. Consider the update rule used in Algorithm~\ref{alg:weighted-vareps-unknown-search}:
    \begin{equation}\label{eq:beta-weighted-update}  
        v_{i} = \argmin_{v \in \partial V_{i-1}} \left\{ \beta d_{G}(v_{i-1},v) + f(v)\right\},
    \end{equation}
    and assume $\beta > 0$ satisfies $\beta < 1-\varepsilon$. Then, for every iteration $i \in [T]$, the node $v_i$ visited by the algorithm on the $i^{th}$ iteration satisfies
    \[
        d_{G}(v_i, g) \leq \frac{1+\varepsilon + \beta}{1 - (\varepsilon + \beta)}\opt.
    \]
\end{lemma}
\begin{proof}[{\bf Proof of Lemma~\ref{lemma:beta-weighted-update-bounded-dist-to-g}}]
    Let $r_i$ be the first vertex outside of $V_{i-1}$ encountered when traversing $P_G(r,g)$ a shortest path from the root $r$ to $g$. Note that $r_i \in \partial V_{i-1}$, and thus by Equation~\eqref{eq:beta-weighted-update},
    \[
        \beta d_{G}(v_{i-1}, v_{i}) + f(v_{i}) \leq \beta d_{G}(v_{i-1}, r_i) + f(r_i).
    \]
    In particular, by the triangle inequality we can bound
    \begin{align*}
        f(v_{i}) &\leq \beta \left(d_{G}(v_{i-1}, r_i) - d_{G}(v_{i-1}, v_{i}) \right) + f(r_i)\\
        &\leq \beta d_{G}(r_i, v_{i}) + f(r_i)\\
        &\leq \beta\left(d_G(r_i, g) + d_G(v_{i},g)\right) + f(r_i).
    \end{align*}
    Given $r_i \in P_G(r,g)$, $d_G(r_i,g) \leq d_G(r,g) = \opt$. Moreover, recalling that the predictions $f$ must satisfy Equation~\eqref{eq:multiplicative_error_setup}, let $\varepsilon_{v_{i}}$ and $\varepsilon_{r_i}$ be the relative errors at vertex $v_{i}$ and $r_i$ respectively (defined as in Equation~\eqref{eq:define_epsilon_v}). Then we can rewrite the bound above as
    \[
        (1+\varepsilon_{v_{i}}) d_G(v_{i},g) \leq \beta\left(d_G(r_i, g) + d_G(v_{i},g)\right) + (1+\varepsilon_{r_i}) d_G(r_i,g),
    \]
    and hence:
    \[
        (1+\varepsilon_{v_{i}} - \beta) d_G(v_{i}, g) \leq (\beta + 1+ \varepsilon_{r_i})\opt.
    \]
    Using the fact that $\varepsilon_{v_{i}}, \varepsilon_{r_i} \in [-\varepsilon, \varepsilon]$ and the assumption that $\beta$ satisfies $\beta < 1-\varepsilon$, we can use the above bound to conclude
    \[
        d_G(v_{i}, g) \leq \frac{1+ \varepsilon + \beta}{1-(\varepsilon + \beta)}\opt.
    \]
    In particular, this holds on any iteration independently of $i$, so we obtain the desired result.
\end{proof}

\begin{lemma}\label{lemma:beta-weighted-alg-guarantee}
    Consider the exploration problem on $G$ a weighted, undirected graph with predictions $f$ satisfying Equation~\eqref{eq:multiplicative_error_setup}. Assume $\beta > 0$ satisfies
    \begin{equation}\label{eq:beta-weight-constraints}
        \frac{1+\varepsilon}{2} < \beta < 1-\varepsilon.
    \end{equation}
    Then the distance traversed by  Algorithm~\ref{alg:weighted-vareps-unknown-search} on the $i^{th}$ iteration is bounded by 
    \[
        d_{G_i}(v_{i-1}, v_{i}) \leq \frac{\beta}{2\beta - 1 - \varepsilon} \Delta_i + \frac{2\varepsilon}{2\beta -1 -\varepsilon} d_{G}(v_{i}, g).
    \]
\end{lemma}

\begin{proof}[{\bf Proof of Lemma~\ref{lemma:beta-weighted-alg-guarantee}}]
    Let $P_G(u,v)$ denote an (arbitrary) shortest path between nodes $u$ and $v$ in $G$. Algorithm~\ref{alg:weighted-vareps-unknown-search} follows a shortest path in $G_i$ from $v_{i-1}$ to $v_{i}$.  Let $w_i$ be the first vertex outside of $V_{i-1}$ encountered when traversing $P_{G}(v_{i-1}, g)$. Note that as a consequence, $w_i \in \partial V_{i-1}$, and hence by \eqref{eq:beta-weighted-update} we have
    \[
        \beta d_{G_i}(v_{i-1}, v_{i}) + f(v_{i}) \leq \beta d_{G_i} (v_{i-1}, w_i) + f(w_i).
    \]
    Since $w_i \in P_{G}(v_{i-1}, g)$, 
    \[
        d_{G_i}(v_{i-1}, w_i) = d_{G}(v_{i-1}, w_i) = d_{G}(v_{i-1}, g) - d_{G}(w_i, g).
    \]
    Re-arranging and using the above fact, we obtain
    \begin{align*}
        \beta d_{G_i}(v_{i-1}, v_{i}) &\leq \beta d_{G_i}(v_{i-1},w_i) + f(w_i) - f(v_{i})\\
        &= \beta \big(d_{G}(v_{i-1}, g) - d_{G}(w_i, g)\big) + f(w_i) - f(v_{i})\\
        &=\beta \big(d_{G}(v_{i-1}, g) - d_{G}(v_{i}, g)\big) + \big(\beta d_{G}(v_{i}, g) - f(v_{i})\big) + \big(f(w_i) - \beta d_{G}(w_i, g)\big).
    \end{align*}
    Let $\varepsilon_{v_{i}}, \varepsilon_{w_i}$ be the relative prediction errors at $v_i$ and $w_i$ respectively (defined as in Equation~\eqref{eq:define_epsilon_v}), and recall the definition of $\Delta_i$ in \eqref{eq:defn-Delta}. Then we can rewrite the above as
    \[
        \beta d_{G_i}(v_{i-1}, v_{i}) \leq \beta \Delta_i + (\beta - 1 - \varepsilon_{v_{i}})d_{G}(v_{i},g) + (1+\varepsilon_{w_i} - \beta) d_{G} (w_i, g).
    \]
    Because $w_i \in P_{G}(v_{i-1}, g)$, we have that
    \[
        d_{G}(v_{i-1}, w_i) + d_{G}(w_i, g) = d_{G}(v_{i-1}, g) \leq d_{G}(v_{i-1}, v_{i}) + d_{G}(v_{i}, g).
    \]
    Moreover, by upper bound on  $\beta$ in \eqref{eq:beta-weight-constraints}, $(1+\varepsilon_{w_i} - \beta) \geq 0$, so we can revise our upper bound:
    \[
        \beta d_{G_i}(v_{i-1}, v_{i}) \leq \beta \Delta_i + (\beta - 1 - \varepsilon_{v_{i}})d_{G}(v_{i},g) + (1+\varepsilon_{w_i} - \beta)\big( d_{G}(v_{i-1}, v_{i}) + d_{G}(v_{i}, g)\big).
    \]
    Re-arranging and recalling $\varepsilon_{v_{i-1}}, \varepsilon_{w_i}\in [-\varepsilon, \varepsilon]$ yields
    \[
        (2\beta-1-\varepsilon) d_{G_i}(v_{i-1}, v_{i}) \leq \beta \Delta_i  + 2\varepsilon d_{G}(v_{i}, g)
    \]
    Leveraging the lower bound on $\beta$ in \eqref{eq:beta-weight-constraints}, we can divide to obtain the desired result.
\end{proof}

Theorem~\ref{thm:exploration-trees-epsilon-unknown} follows immediately from Lemmas~\ref{lemma:beta-weighted-update-bounded-dist-to-g} and \ref{lemma:beta-weighted-alg-guarantee}:
    \begin{proof}[{\bf Proof of Theorem~\ref{thm:exploration-trees-epsilon-unknown}}] 
        Observe that for $\varepsilon \in (0, 1/3)$, $\beta = 2/3$ always satisfies Equation~\eqref{eq:beta-weight-constraints}, independently of the value of $\varepsilon$. Thus for this setting, Lemma~\ref{lemma:beta-weighted-alg-guarantee} implies that the update cost on a single iteration of Algorithm~\ref{alg:weighted-vareps-unknown-search} is bounded as
        \begin{equation}\label{eq:generic-bound-abc}
            d_{G_i}(v_{i-1}, v_{i}) \leq \frac{2}{1-3\varepsilon}\Delta_i + \frac{6\varepsilon}{1-3\varepsilon}d_{G}(v_{i},g).
        \end{equation}

        In particular, for $G$ a weighted tree, on all iterations $i$, for all $u, v\in V_i \cup \partial V_i$, 
        \[
            d_{G_i}(u,v) = d_{G}(u,v).
        \]
        This, in combination with the choice of $\beta = 2/3$ implies that Lemma~\ref{lemma:beta-weighted-update-bounded-dist-to-g} applies, so for all vertices $v_i$ visited by the algorithm, we have
        \begin{equation}\label{eq:UB-on-v_i-dist-in-proof-of-tree-exploration}
            d_{G}(v_{i},g) \leq \frac{5+3\varepsilon}{1-3\varepsilon}\opt.
        \end{equation}
        We can then upper bound the last term in Equation~\eqref{eq:generic-bound-abc} and obtain:
        \[
            d_{G_i}(v_{i-1}, v_{i}) \leq \frac{2}{1-3\varepsilon}\Delta_i + \frac{6\varepsilon}{1-3\varepsilon}\cdot \frac{5+3\varepsilon}{1-3\varepsilon}\opt
        \]
        Thus, letting $T$ denote the total number of iterations made by Algorithm~\ref{alg:weighted-vareps-unknown-search}, summing over all iterations yields
        \[
            \alg = \sum_{i\in [T]}d_{G_i}(v_{i-1}, v_{i}) \leq \frac{2}{1-3\varepsilon}\sum_{i\in [T]}\Delta_i + \frac{6\varepsilon}{1-3\varepsilon}\cdot \frac{5+3\varepsilon}{1-3\varepsilon}\opt \cdot T.
        \]
        In particular, $\sum_{i\in [T]}\Delta_i = \opt$, and as every iteration $i$ must end at some distinct $v_i$ satisfying \eqref{eq:UB-on-v_i-dist-in-proof-of-tree-exploration},
        \[
            T \leq \left|B\left(g, \frac{5+3\varepsilon}{1-3\varepsilon}\opt\right)\right|\leq n.
        \]
        We then have:
        \begin{align*}            
            \alg &\leq \frac{2}{1-3\varepsilon}\opt + \frac{6\varepsilon}{1-3\varepsilon}\cdot \frac{5+3\varepsilon}{1-3\varepsilon}\opt \cdot T\\
            &\leq \opt\left(2+\frac{6\varepsilon}{1-3\varepsilon} + \frac{6\varepsilon}{1-3\varepsilon}\cdot \frac{5+3\varepsilon}{1-3\varepsilon}\cdot n\right).
        \end{align*}
        Giving the result in the statement of the theorem.
\end{proof}

\subsubsection{Lower bounds For Exploration With Relative Error}

\begin{proof}[{\bf Proof of Theorem~\ref{theorem:multiplicative_error_explo_LB}}]

    We consider a star in which every edge has weight $w_1$, to this, we add a new vertex $g$ connected to one of the outside vertices by an edge of weight $w_2$. We consider the case when the initial position $r$ is the central node of the star. In this case, the optimal algorithmic cost is
    \[
        \opt = d_G(r,g) = w_1 + w_2.
    \]
    For the given $\varepsilon \in (0,1)$, consider choice of $w_1$ and $w_2$ such that 
    \[
        \varepsilon = \frac{w_1}{w_1 + w_2}.
    \]
    Note that in particular, such a setting of weights allows for the following predictions: every node except for the root and the goal can have prediction
    \[
        f(v) = (1-\varepsilon)(2w_1 + w_2).
    \]
    For the above setting of $w_1$ and $w_2$, this satisfies Equation~\ref{eq:multiplicative_error_setup} with respect to $\varepsilon$. In particular, for a searcher starting at the root, all neighbors of the root appear identical.
    
     In this error regime, every algorithm for the exploration problem has to explore all branches of the star before finding $g$ in the worst-case. Thus any algorithm has to incur cost at least
     \[
        \alg  = 2 w_1 \cdot (n-3) + (w_1 + w_2) = w_1 ( 2(n-3) + 1) + w_2.
     \]
     The competitive ratio in this instance is thus
     \[
        \frac{\alg}{\opt} = \left(2(n-3) + 1\right)\frac{w_1}{w_1 + w_2} + \frac{w_2}{w_1 + w_2} = \left(2(n-3) + 1\right)\varepsilon + (1-\varepsilon) = \Theta(1+ n\varepsilon).
     \]
\end{proof}

\begin{proof}[\bf Proof of \Cref{prop:randomized-multiplicative-lower-bound}]
    Consider a distribution over instances obtained by taking the instance defined in the proof of \Cref{theorem:multiplicative_error_explo_LB}, selecting a neighbor $v$ of the root vertex $r$ uniformly at random, and replacing the edge incident to the goal vertex $g$ with an edge $gv$ of weight $w_2$.

    Any deterministic algorithm for the exploration problem running on an instance sampled from this distributions incurs expected cost at least:
    \[
        \alg = (n-3)\cdot {1 \over 2} 2w_2 +\opt =  \left((n-3)\varepsilon +1\right)\opt \geq  \left(1+{n \varepsilon \over 2 }\right)\opt.
    \]
    The lower bound for randomized algorithms then follows from applying Yao's minimax principle~\citep{yao1977probabilistic}.
\end{proof}

\subsection{Proofs For Section~\ref{sec:planning}}

Throughout this subsection, we will use the following properties of $\phi_0$ and $\phi_1$ established by \citet{banerjee2022graph}:
\begin{lemma}\label{lemma:bannerjee-phi-bounds}{Corollary 5.4 and Lemma 5.10 in \citet{banerjee2022graph}}.
    Given $G$ an unweighted graph, for any $u,v\in G$,
    \[
        \phi_0(u)+\phi_0(v) \geq d_{G}(u,v).
    \]
    For $G$ weighted, 
    \[
        \phi_1(u) + \phi_1(v) \geq 2d_{G}(u,v).
    \]
\end{lemma}

\subsubsection{Planning Bounds Via Metric Embeddings}

The distortion of an embedding can be related to its Lipschitz constant and that of its inverse: the Lipschitz constant of $\tau$ is defined as:
\[
    \Lipnorm{\tau} \defeq \max_{x_1, x_2 \in X} { d_Y(\tau(x_1), \tau(x_2)) \over d_X(x_1,x_2)}.
\]
Note that any map with non-trivial distortion must be injective, and thus considering $\tau^{-1}:Y\rightarrow X$,
\[
    \operatorname{dist}(\tau) = \Lipnorm{\tau}\cdot\Lipnorm{\tau^{-1}}.
\]

To prove Lemma~\ref{lemma:planning-embedding}, we'll use the following fact to relate tours in $G$ to tours in some embedding. 
\begin{lemma}\label{lemma:tours-via-embedding}
    Consider an embedding $\tau:G\rightarrow G'$ for $G = (V,E)$ and $G' = (V',E')$. Then for any $S \subseteq V$, 
    \[
        \tour_G(S) \leq \Lipnorm{\tau^{-1}} \cdot \tour_{G'}(\tau(S)).
    \]
\end{lemma}

\begin{proof}[{\bf Proof of Lemma~\ref{lemma:tours-via-embedding}}]
    Recall the definition of $\tour_G(S)$ given in Equation~\eqref{eq:defn_tour}:
    \[
        \tour_{G}(S) \defeq \max_{v\in S}\min_{W\in\mathcal{W}(v,S)} \text{length}_G(W),
    \]
    where $\mathcal{W}(v,S)$ is the set of walks in $G$ starting at vertex $v$ and visiting every vertex in $S$. Consider any walk $\cW = (u_1, ... , u_k)$ in $G'$ starting at some $u_1 \in \tau(S)$ and visiting all of $\tau(S)$. Let $W' = (u'_1, ... , u'_{k'})$ be the subsequence of $W$ containing only the points in $\tau(S)$. Note that $W'$ contains all of the points in $\tau(S)$.  We have:
    \begin{align}
        \text{length}_{G'}(W) &= \sum_{i=1}^{k-1} d_{G'} (u_i, u_{i+1}) \geq \sum_{i=1}^{k'-1} d_{G'} (u'_i, u'_{i+1}) \\
        &\geq \sum_{i=1}^{k'-1} {1\over \Lipnorm{\tau^{-1}}} \cdot d_{G} (\tau^{-1}(u'_i), \tau^{-1}(u'_{i+1})).
    \end{align}

    So, letting $W''$ be the walk visiting the vertices $(\tau^{-1}(u_i'))_{i=1}^{k'}$ in order while walking the shortest path in $G$ between them. We then have:
    \[
        \text{length}_{G'}(W) \geq \sum_{i=1}^{k'-1} \frac{1}{\Lipnorm{\tau^{-1}}}\cdot  d_{G} (\tau^{-1}(u'_i), \tau^{-1}(u'_{i+1})) = {1 \over \Lipnorm{\tau^{-1}}} \cdot \text{length}_G(W'').
    \]
    In particular, for any starting point $v \in S$ and any walk in $W \in \cW (\tau(v), \tau(S))$ there exists some walk $W'' \in \cW(v, S)$ such that:
    \[
        \text{length}_G(W'') \leq \Lipnorm{\tau^{-1}}\cdot \text{length}_{G'}(W),
    \]
    so that, for every $v \in S$:
    \[
        \min_{W \in \mathcal{W}(v,S)} \text{length}_G(W) \leq \Lipnorm{\tau^{-1}}\min_{W\in \cW(\tau(v), \tau(S))} \text{length}_{G'}(W) \leq \Lipnorm{\tau^{-1}}\max_{u\in \tau(S)}\min_{W\in \cW(u, \tau(S))} \text{length}_{G'}(W),
    \]
    and hence:
    \[
         \max_{v\in S}\min_{W \in \mathcal{W}(v,S)} \text{length}_G(W) \leq \Lipnorm{\tau^{-1}}\max_{u\in \tau(S)}\min_{W\in \cW(u, \tau(S))} \text{length}_{G'}(W),
    \]
    completing the proof.
\end{proof}

We now use Lemma~\ref{lemma:tours-via-embedding} to establish Lemma~\ref{lemma:planning-embedding}.

\begin{proof}[{\bf Proof of Lemma~\ref{lemma:planning-embedding}.}]
    Given a real-valued function $f:V\rightarrow \R$, we denote the sublevel set of $f$ about threshold $c$ as
    \[
        L^{-}_{f}(c) \defeq \{v\in V: f(v) \leq c\}.
    \]
    
    For the first part of the result, let $G$ be an unweighted graph and consider the sublevel set $L^{-}_{\phi_0}(\lambda)$. By definition of the Lipschitz constant of $\tau:G\rightarrow G'$, for all $u, v \in G$
    \[
        d_{G'}(\tau(u),\tau(v)) \leq \Lipnorm{\tau} d_{G}(u,v).
    \]
    Thus by Lemma~\ref{lemma:bannerjee-phi-bounds}, the embedding of the sublevel set has bounded diameter: let $u, v\in L^{-}_{\phi_0}(\lambda)$ such that $\text{diam}\big(\tau(L^{-}_{\phi_0}(\lambda))\big) = d_{G'}(\tau(u),\tau(v))$. Then
    \[
        \text{diam}\big(\tau(L^{-}_{\phi_0}(\lambda))\big) =  d_{G'}(\tau(u),\tau(v)) \leq \Lipnorm{\tau} d_{G}(u,v) \leq \Lipnorm{\tau}(\phi_0(u)+\phi_0(v)) \leq \Lipnorm{\tau}\cdot 2\lambda.
    \]
    Using this result along with the bound from Lemma~\ref{lemma:tours-via-embedding} and the assumption that $G'$ is $c_{G'}$ easily-tourable, we can bound
    \begin{align*}
        \tour_{G}(L^{-}_{\phi_0})(\lambda) &\leq \Lipnorm{\tau^{-1}}\cdot\tour_{G'}\big(\tau(L^{-}_{\phi_0}(\lambda)\big)\\
        &\leq \Lipnorm{\tau^{-1}} \cdot c_{G'}\text{diam}\big(\tau(L^{-}_{\phi_0}(\lambda))\big)\\
        &\leq \Lipnorm{\tau^{-1}} \cdot c_{G'}\cdot2\lambda \Lipnorm{\tau}\\
        &=2\rho c_{G'} \lambda.
    \end{align*}
    We now use this bound to analyze the cost of Algorithm~\ref{alg:banerjee-planning}. Algorithm~\ref{alg:banerjee-planning} sequentially visits sublevel sets of $\phi_0$. On an iteration $k$ corresponding to threshold $\lambda_k$, the algorithm visits each node in $L^-_{\phi_0}(\lambda_k) \subseteq V$ by computing a constant-factor approximation to the following problem: for $\{v_1,\dots,v_{\abs{L^-_{\phi_0}(\lambda_k)}}\}$ the nodes of $L^-_{\phi_0}(\lambda_k)$ and $\Pi(n)$ the set of permutations on integers $1,\dots,n$, the algorithm computes
    \[
        \sigma^*_{\lambda_k} \defeq \argmin_{\sigma \in \Pi(|L^-_{\phi_0}(\lambda_k) |)} \sum_{i = 1}^{\abs{L^-_{\phi_0}(\lambda_k) }} d_G(v_{\sigma(i)}, v_{\sigma(i+1)}).
    \]
    The total distance travelled on the iteration $k$ is thus
    \[
        \sum_{i = 1}^{\abs{L^-_{\phi_0}(\lambda_k) }} d_G(v_{\sigma^*(i)}, v_{\sigma^*(i+1)}) \leq \max_{v\in L^-_{\phi_0}(\lambda_k)}\min_{W\in\mathcal{W}(v,L^-_{\phi_0}(\lambda_k))} \text{length}_G(W) = \tour_{G}(L^-_{\phi_0}(\lambda_k)).
    \]
    Algorithm~\ref{alg:banerjee-planning} begins with $\lambda_0 = 1$ and doubles the threshold on each iteration, such that $\lambda_k = 2^{k}$. In particular, $\phi_0(g) = \cE_0$, so the algorithm is guaranteed to terminate by the time it has visited every node of $L^-_{\phi_0}(\lambda_k)$ for the first sufficiently large threshold $\lambda_k \geq \cE_0$. The algorithmic cost can thus be bounded as
    \begin{align*}
        \alg &\leq d_{G}(r, L^{-}_{\phi_0}(1)) + \sum_{k = 0}^{\lceil \log_2(\cE_0)\rceil} \tour_{G}(L^{-}_{\phi_0}(2^k))\\
        &\leq d_{G}(r, L^{-}_{\phi_0}(1)) + 2\rho c_{G'} \sum_{k = 0}^{\lceil \log_2(\cE_0)\rceil} 2^k\\
        &\leq d_{G}(r, L^{-}_{\phi_0}(1)) + 2\rho c_{G'}(4\cE_0 - 1).
    \end{align*}
    Using Lemma~\ref{lemma:bannerjee-phi-bounds}, we can bound the transition cost from $r$ to the first sublevel set as
    \[
        d_{G}(r, L^{-}_{\phi_0}(1)) \leq d_{G}(r, g) + d_{G}(g, L^{-}_{\phi_0}(1)) \leq \opt + (\phi_0(g) + 1) = \opt + \cE_0 + 1.
    \]
    Combining these yields
    \[
        \alg \leq \opt + \cE_0(8\rho c_{G'} + 1) -2\rho c_{G'} + 1 = \opt + O(\rho c_{G'}\cE_0),
    \]
    as desired.

    For the second part of the result, let $G$ be an graph with integer-valued distances consider the sublevel set $L^{-}_{\phi_1}(\lambda)$, and note that Lemma~\ref{lemma:tours-via-embedding} holds for both weighted and unweighted graphs.
    The proof then follows analogously to the above argument, using the appropriate bound relating $\phi_1$ to distances in $G$ from Lemma~\ref{lemma:bannerjee-phi-bounds}.
\end{proof}

\subsubsection{Lower Bounds For Planning Problems}
\begin{figure}[H]
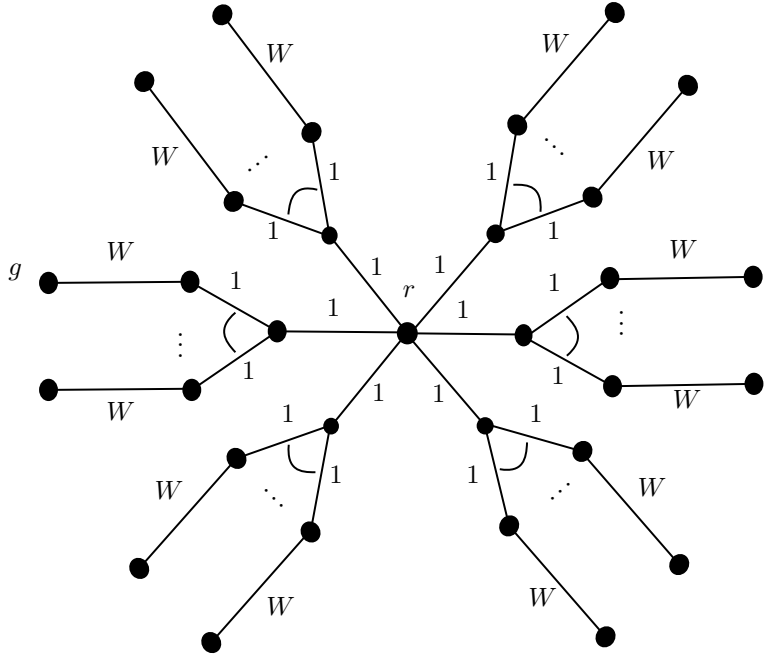

    \centering
    \include{img/PlanningE1LowerboundsLarge}
    \caption{Reproduction of Figure~\ref{fig:planning-E1-lowerbound}. The lower bound construction for the proof of Lemma~\ref{lemma:lowerbound-E1-planning}}
    \label{fig:planning-E1-lowerbound-reproduction}
\end{figure}
\begin{proof}[{\bf Proof of Lemma~\ref{lemma:lowerbound-E1-planning}}]
    We construct a family of graphs with uniform predictions and analyze the worst-case cost incurred by any algorithm for the planning problem.

    For a given $\Delta$ and $W \in \N$ let $G_{\Delta, W}$ be following graph: consider a root node $r$ with $\Delta$ child nodes $v_1,\dots,v_{\Delta}$. Let every edge $(r, v_i)$ have edge weight $1$. Each child node $v_i$ then has $\Delta-1$ descendants $u^{i}_1,\dots,u^{i}_{\Delta-1}$ with edge weights $1$ for each edge $(v_i, u^{i}_j)$. Each of these descendants has a single child node $w^{i}_j$ to which $u^{i}_j$ is attached with edge weight $W$. This construction is illustrated in Figure~\ref{fig:planning-E1-lowerbound}.

    We consider the planning problem when the searcher is initialized at the root node $r$ described above, and the goal is the leaf node $w^{1}_1$. We consider the case when error in the predictions is such that each subtree rooted at $v_i$ appears to have the same predictions. In this construction, predictions are equal to true distance-to-goal for all nodes which are descendents of $v_1$ for $i\neq 1$, and error is allocated only over descendents of $v_1$. We first calculate the total error in such predictions: 
    \[
        \cE_1 = |f(v_1)-d_G(v_1,g)| + |f(g)| + \sum_{j=2}^{\Delta-1} \abs{f(u^{1}_j) - d_G(u^{1}_j, g)}+\abs{f(w^{1}_j) - d_G(w^{1}_j, g)} = 2W + 6\Delta + 4.
    \]

    Under these predictions, all nodes on a given level from the root appear identical to the searcher. As a result, in the worst-case any algorithm for this problem must visit every node and incur cost $\alg$ at least as large as the shortest tour of the graph starting at $r$ and ending at $g$. Hence:
    \[
        \alg \geq W(2\Delta^2-2\Delta-1) + 2(\Delta^2-1).
    \]
    On the other hand, we have $\opt =d_{G}(r,g) = W + 2$. This gives:
    \[
        \alg - \opt \geq 2W(\Delta^2-\Delta-1) + 2(\Delta^2-2).
    \]
    In order to understand how this cost scales with our parameters of interest, we now establish bounds on the doubling constant of $G$ that show that $\lambda = \Theta(\Delta^2)$. Recall that the doubling constant is defined to be the minimum value of $\lambda$ such that for any radius $R$, any ball of radius $R$ can be covered with at most $\lambda$ many balls of radius $R/2$. Observe that the number of nodes $n$ always upper bounds the doubling constant, so $\lambda \leq 2\Delta^2 + \Delta + 1$. To lower bound the doubling constant, consider the ball of radius $W + 2$ centered at the root node $r$, and note that this ball contains the entire graph. For $W$ large ($W \geq 4$), $\Delta^2 + 1$ many balls of radius $W/2 + 1$ are required to cover the nodes of the graph, hence $\Delta^2-1\leq \lambda$. 
    
    The fact that $a\geq 1$ then follows from observing that $\lambda$ is independent of $W$ while $\alg - \opt = \Omega(W\Delta^2)$. The fact that if $a=1$, $b\geq 1$ similarly follows from the above along with $\cE_1 = \Theta(W +\Delta)$.  
\end{proof}

\begin{proof}[{\bf Proof of Lemma~\ref{lemma:lowerbound-E1-planning-trees}}]
    To establish that $a\geq 1$ and that $a + b\geq 3$, we consider the same construction outlined in the proof of Lemma~\ref{lemma:lowerbound-E1-planning} above. The same argument implies $a\geq 1$. Setting $W = \Delta$ yields a family of problem instances on integer-weighted trees for which $\alg - \opt = \Omega(\Delta^3)$ where $\cE_1 = \Theta(\Delta)$. Thus on this family of instances any algorithm which is guaranteed to incur cost $\alg - \opt = O(\cE_1^{a}\Delta^{b})$ must have $a+b\geq 3$.

    To establish that $b\geq 1$, we consider a construction in which $\cE_1$ scales independently of $\Delta$. Consider the weighted star with edge weights $w$, and assume the searcher is initialized at the central root node $r$ and the goal node $g$ is a leaf node as illustrated on the right side of Figure~\ref{fig:Lowerbound1}. We consider the same predictions constructed in the proof of Theorem~\ref{thm:optimal_E1_lowerbounds}: all nodes except for $g$ have $f(v) = d_{G}(v,g)$, and $f(g) = 2w$ so that predictions at all leaf nodes appear uniform. Then in the worst case, the searcher must visit every node in the graph, incurring traversal cost
    \[
        \alg = w (2\Delta -1).
    \]
    However, the total $\ell_1$ norm of the vector of errors is $\cE_1 = 2w$ independent of $\Delta$, and $\opt = w$, so the result follows.
\end{proof}

\begin{lemma}\label{lemma:lowerbound-E1-planning-rho}
    Let $\mathcal{A}$ be any algorithm for the planning problem on weighted trees which is guaranteed to incur cost: $\opt + O(\cE^{a}_1 \rho^{b})$ on graph searching instance $\mathcal{I}$, where $\rho$ is the minimum distortion of embedding the instance graph into the path. Then $a\geq 1$ and $b\geq 1$.
\end{lemma}
\begin{proof}
    The proof is entirely analogous to that of the second part of Lemma~\ref{lemma:lowerbound-E1-planning-trees}. For any value of $\cE_1$, one can make use of the same construction (the weighted star, with $r$ the central node and $g$ a leaf) with weights $w = \cE_1/2$ on each edge.
    
    The result then follows from observing that for the family of constructed graphs, $\rho = \Delta$ and the analogous calculations.
\end{proof}

\subsection{Planning On Trees With Integer-Valued Distances}

In this section, we prove Lemma~\ref{lemma:planning-on-trees}. The result follows from arguments analogous to those outlined in \citet{banerjee2022graph} Section 5.1: we first state and prove three necessary lemmas. 
\begin{lemma}[Analogous to Lemma 5.3 in \citet{banerjee2022graph}]\label{lem:integer_e1_cardinality_bound}
    Given $G$ with positive integer distances $d:V\times V \rightarrow \mathbb{Z}_{\geq 0}$, for any $U\subseteq V$
    \[
        \abs{S \setminus M(U)} \leq \sum_{u\in U} \varphi_1(u),
    \]
    where
    \[
        M(U) \defeq \{v\in V: d(v,u) = d(v,u') \ \forall u, u'\in U\}.
    \]
\end{lemma}
\begin{proof}[{\bf Proof Lemma~\ref{lem:integer_e1_cardinality_bound}}]
    Given $d:V\times V \rightarrow \mathbb{Z}_{\geq 0}$, for any $U\subseteq V$, $\forall w\not\in M(U)$ let $u_w, v_w$ denote elements of $U$ such that $d_G(w,u_w)\neq d_G(w,v_w)$. In particular, because $d(\cdot,\cdot)$ is integer-valued, $\forall w\not\in M(U)$
    \[
        \abs{d_{G}(w,u_w)- d_{G}(w,v_w)} \geq 1.
    \]
    In particular, for any $S \subseteq V$
    \begin{align*}
        \sum_{u\in U} \varphi_1(u) &=\sum_{u\in U}\sum_{w\in V}  \abs{f(w) - d_G(w,u)} \\
        &\geq \sum_{u\in U}\sum_{w\in S\setminus M(U)} \abs{f(w) - d_{G}(w, u)}\\
        &\geq  \sum_{w\in S\setminus M(U)} \abs{f(w) - d_{G}(w, u_w)}+\abs{f(w) - d_{G}(w, v_w)} \\
        &\geq \sum_{w\in S\setminus M(U)} \abs{d_{G}(w, u_w)- d_{G}(w, v_w)}\\
        &\geq \abs{S\setminus M(U)}.
    \end{align*}
\end{proof}

\begin{lemma}[Generalization of Lemma 5.10 in \citet{banerjee2022graph}]\label{lem:sum-of-phis-is-greater-than-distance}
For any $u, v\in V$, we have:
    \[
        \varphi_1(u) + \varphi_1(v) \geq 2 d(u,v) + \sum_{w \in V \setminus \{u,v\}} \left| d(u,w) - d(v,w)\right|.
    \]
\end{lemma}
\begin{proof}[{\bf Proof of Lemma~\ref{lem:sum-of-phis-is-greater-than-distance}}]
The proof is a straight-forward application of the triangle inequality:
    \begin{align*}
        \varphi_1(u) + \varphi_1(v) &= \sum_{w \in V} \left|d(u,w) - f(w)\right| + \left|d(v,w) - f(w)\right|\\
        &\geq \sum_{w \in V} \left|d(u,w) - f(w) - d(v,w) + f(w)\right|\\
        &= \sum_{w \in V} \left|d(u,w) - d(v,w)\right|\\
        &= 2 d(u,v) + \sum_{w \in V \setminus \{u,v\}} \left| d(u,w) - d(v,w)\right|.
    \end{align*}
\end{proof}

We also utilize the following bound on the size of the minimum Steiner tree of any sublevel set of $\phi_1$: recall that for a real-valued function $f:V\rightarrow \R$, we denote the sublevel set of $f$ about threshold $c$ as
    \[
        L^{-}_{f}(c) \defeq \{v\in V: f(v) \leq c\}.
    \]
\begin{lemma}\label{lem:steiner_tree_bound}
    For $G$ a connected tree with at least three nodes, integer edge weights, and maximum degree $\Delta$, let $C_\lambda$ denote the set of vertices in the minimum Steiner tree containing all vertices in $L^{-}_{\phi_1}(\lambda)$. Then
    \[
        |C_\lambda| \leq \lambda \Delta.
    \]
\end{lemma}
\begin{proof}[{\bf Proof of Lemma~\ref{lem:steiner_tree_bound}}]
    By definition, $L^{-}_{\phi_1}(\lambda) \subseteq C_{\lambda}$. Let $u_1, u_2 \in L^{-}_{\phi_1}(\lambda)$ such that
    \[
        d(u_1, u_2) = \text{diam}(L^{-}_{\phi_1}(\lambda))
    \]
    If $\not\exists w \in C_\lambda$ such that $d(u_1,w) = d(u_2,w)$, then Lemma~\ref{lem:integer_e1_cardinality_bound} implies
    \[
        \abs{C_\lambda} = \abs{C_\lambda \setminus M({u_1,u_2})} \leq \varphi_1(u_1) + \varphi_1(u_2) \leq 2\lambda.
    \]
    In particular, for $G$ connected with at least three nodes, $\Delta \geq 2$, so the desired bound holds.

    We now consider the case when $\exists w\in C_\lambda$ such that $d(u_1,w) = d(u_2,w)$. Let $q_1,\dots,q_k$ denote the neighbors of $w$ and let $T_{i} \subseteq C_{\lambda}$ denote the subtree of descendants of $w$ that contains $q_i$. Assume  without loss of generality that $T_1 \ni u_1$ and $T_2 \ni u_2$. Note that, because $C_\lambda$ is defined to be minimal,  $\forall i\in [k]$ $L^{-}_{\phi_1}(\lambda)\cap T_i \neq \emptyset$. Let $u_3, ... , u_k$ be points such that $u_i \in L^{-}_{\phi_1}(\lambda) \cap T_i$. Consider any $x\in C_\lambda\setminus\{w\}$. Then $\exists j \in [k]$ such that $x\in T_{j}$. This case is illustrated in Figure~\ref{fig:steiner_tree_bound_image}. 

    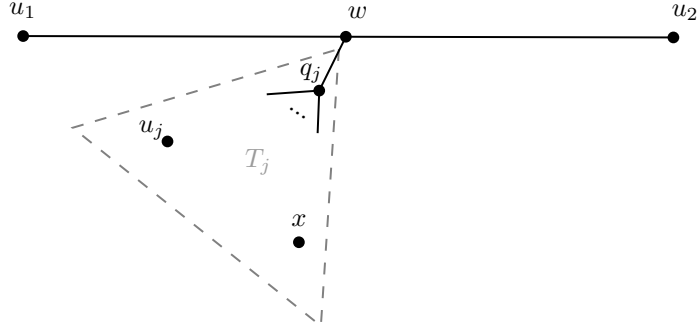
\begin{figure}
        \centering
        \tikzset{every picture/.style={line width=0.75pt}} %

\begin{tikzpicture}[x=0.75pt,y=0.75pt,yscale=-1,xscale=1,scale=0.7]
\draw    (100,51) -- (572,52) ;
\draw  [fill={rgb, 255:red, 0; green, 0; blue, 0 }  ,fill opacity=1 ] (100,51) .. controls (100,49.07) and (101.57,47.5) .. (103.5,47.5) .. controls (105.43,47.5) and (107,49.07) .. (107,51) .. controls (107,52.93) and (105.43,54.5) .. (103.5,54.5) .. controls (101.57,54.5) and (100,52.93) .. (100,51) -- cycle ;
\draw  [fill={rgb, 255:red, 0; green, 0; blue, 0 }  ,fill opacity=1 ] (568.5,52) .. controls (568.5,50.07) and (570.07,48.5) .. (572,48.5) .. controls (573.93,48.5) and (575.5,50.07) .. (575.5,52) .. controls (575.5,53.93) and (573.93,55.5) .. (572,55.5) .. controls (570.07,55.5) and (568.5,53.93) .. (568.5,52) -- cycle ;
\draw  [fill={rgb, 255:red, 0; green, 0; blue, 0 }  ,fill opacity=1 ] (332.5,51.5) .. controls (332.5,49.57) and (334.07,48) .. (336,48) .. controls (337.93,48) and (339.5,49.57) .. (339.5,51.5) .. controls (339.5,53.43) and (337.93,55) .. (336,55) .. controls (334.07,55) and (332.5,53.43) .. (332.5,51.5) -- cycle ;
\draw  [fill={rgb, 255:red, 0; green, 0; blue, 0 }  ,fill opacity=1 ] (313.33,90.33) .. controls (313.33,88.4) and (314.9,86.83) .. (316.83,86.83) .. controls (318.77,86.83) and (320.33,88.4) .. (320.33,90.33) .. controls (320.33,92.27) and (318.77,93.83) .. (316.83,93.83) .. controls (314.9,93.83) and (313.33,92.27) .. (313.33,90.33) -- cycle ;
\draw  [color={rgb, 255:red, 128; green, 128; blue, 128 }  ,draw opacity=1 ][dash pattern={on 4.5pt off 4.5pt}] (330.99,60.37) -- (318,260) -- (139.15,117.13) -- cycle ;
\draw  [fill={rgb, 255:red, 0; green, 0; blue, 0 }  ,fill opacity=1 ] (298.67,199.67) .. controls (298.67,197.73) and (300.23,196.17) .. (302.17,196.17) .. controls (304.1,196.17) and (305.67,197.73) .. (305.67,199.67) .. controls (305.67,201.6) and (304.1,203.17) .. (302.17,203.17) .. controls (300.23,203.17) and (298.67,201.6) .. (298.67,199.67) -- cycle ;
\draw  [fill={rgb, 255:red, 0; green, 0; blue, 0 }  ,fill opacity=1 ] (204,127) .. controls (204,125.07) and (205.57,123.5) .. (207.5,123.5) .. controls (209.43,123.5) and (211,125.07) .. (211,127) .. controls (211,128.93) and (209.43,130.5) .. (207.5,130.5) .. controls (205.57,130.5) and (204,128.93) .. (204,127) -- cycle ;
\draw    (336,51.5) -- (316.83,90.33) ;
\draw    (316.83,90.33) -- (279,93) ;
\draw    (316.83,90.33) -- (315.67,121) ;

\draw (261.33,130.33) node [anchor=north west][inner sep=0.75pt]  [color={rgb, 255:red, 155; green, 155; blue, 155 }  ,opacity=1 ] [align=left] {$\displaystyle T_{j}$};
\draw (300.67,69) node [anchor=north west][inner sep=0.75pt]   [align=left] {$\displaystyle q_{j}$};
\draw (335.33,27.33) node [anchor=north west][inner sep=0.75pt]   [align=left] {$\displaystyle w$};
\draw (91.33,24.67) node [anchor=north west][inner sep=0.75pt]   [align=left] {$\displaystyle u_{1}$};
\draw (568.67,26) node [anchor=north west][inner sep=0.75pt]   [align=left] {$\displaystyle u_{2}$};
\draw (294.67,177.33) node [anchor=north west][inner sep=0.75pt]   [align=left] {$\displaystyle x$};
\draw (184.67,110) node [anchor=north west][inner sep=0.75pt]   [align=left] {$\displaystyle u_{j}$};
\draw (295,98) node [anchor=north west][inner sep=0.75pt]  [rotate=-36.22] [align=left] {...};

\end{tikzpicture}
        \vspace{-10mm}
        \caption{Visual aid for proof of Lemma~\ref{lem:steiner_tree_bound}, for the case when $\exists w\in C_\lambda$ such that $d(u_1,w) = d(u_2,w)$.}
        \label{fig:steiner_tree_bound_image}
    \end{figure}
    
    Assume without loss of generality that $j\neq 1$ (this can be assumed WLOG because if $x\in T_1$, then the below argument can be carried out with respect to $u_2$). Because $G$ is a tree and $x\not\in T_1$,
    \[
        d(u_1, x) = d(u_1, w) + d(w, x)
    \]
    By choice of $u_1, u_2$, $\text{diam}(L^{-}_{\phi_1}(\lambda)) = 2d(u_1,w)$ so in particular $d(u_1, w) \geq d(u_j, w) \ \forall u_j \in L^{-}_{\phi_1}(\lambda)$. Thus
    \[
        d(u_1,x) = d(u_1, w) + d(w, x) \geq d(u_j, w) + d(w, x).
    \]
    Moreover, for all $v \in T_j$, $d(v, w) = d(v, q_j) + d(q_j,w)$ by definition of subtree $T_j$, so for non-zero weights,
    \[
        d(u_j, w) + d(w, x) > d(u_j, q_j) + d(x, q_j) \geq d(u_j, x)
    \]
    We thus conclude $d(u_1, x) > d(u_j, x)$, which in particular implies $x\not\in M(\{u_1, \dots, u_k\})$.

    Thus for all $ (C_\lambda \setminus\{w\}) \cap M(\{u_1, \dots, u_k\}) = \emptyset$, so
    \[
        |C_\lambda| -1 = |C_\lambda \setminus M(\{u_1, ... , u_k\})\leq \sum_{i=1}^k \varphi_1(u_i) \leq \lambda \Delta.
    \]
    We have thus established the result in both cases (i.e. when $C_\lambda \cap M(\{u_1, \dots, u_k\}) = \emptyset$ and when $C_\lambda \cap M(\{u_1, \dots, u_k\}) \neq \emptyset)$.
\end{proof}

\begin{proof}[{\bf Proof of Lemma~\ref{lemma:planning-on-trees}}]
    Consider the cost of visiting every node in $C_\lambda$, a minimum Steiner tree containing the sublevel set $L^{-}_{\phi_1}(\lambda)$. Because $C_\lambda$ is minimal, 
    \[
        \text{diam}(C_\lambda) = \text{diam}(L^{-}_{\phi_1}(\lambda)) \leq \lambda
    \]
    where the last inequality follows from Lemma~\ref{lem:sum-of-phis-is-greater-than-distance}. In particular, traversing $C_\lambda$ to visit every node incurs travel cost at most $\text{diam}(C_\lambda)\cdot \abs{C_\lambda}$. Combining the above bound and Lemma~\ref{lem:steiner_tree_bound} implies  $\text{diam}(C_\lambda)\cdot \abs{C_\lambda} \leq \lambda^2 \Delta$.

    Algorithm~\ref{alg:banerjee-planning} with objective $\phi_1$ proceeds by iteratively visiting every node in the sublevel set $L^{-}_{\phi_1}(\lambda)$ by computing and traversing a minimum Steiner tree $C_\lambda$ that contains the sublevel set. Algorithm~\ref{alg:banerjee-planning} begins with $\lambda_0 = 1$ and doubles the threshold on each iteration, such that $\lambda_k = 2^{k}$. In particular, $\phi_1(g) = \cE_1$, so the algorithm is guaranteed to terminate by the time it has visited every node of $L^-_{\phi_1}(\lambda_k)$ for the first sufficiently large threshold $\lambda_k \geq \cE_1$. The algorithmic cost of Algorithm~\ref{alg:banerjee-planning} with objective $\phi_1$ is thus bounded by 
    \begin{align*}
        \alg &= d(r, L^{-}_{\phi_1}(1)) + \sum_{k = 0}^{\lceil \log_2(\cE_1)\rceil} \text{diam}(C_{2^k})\cdot \abs{C_{2^k}}\\
        &\leq d(r, L^{-}_{\phi_1}(1)) + \Delta \sum_{k = 0}^{\lceil \log_2(\cE_1)\rceil} (2^{k})^2\\
        &\leq d(r, L^{-}_{\phi_1}(1)) + \frac{\Delta}{3}(16\cE^2_1 -1)
    \end{align*}
    Additionally, leveraging Lemma~\ref{lem:sum-of-phis-is-greater-than-distance} and the fact that $\phi_1(g) \leq cE_1$,
    \[
        d(r, L^{-}_{\phi_1}(1)) \leq d(r, g) + d(g, L^{-}_{\phi_1}(1)) \leq \opt + \frac{1}{2}\left(\cE_1 + 1\right).
    \]
    Combining these bounds yields the desired result in the regime $\cE_1 \geq 1$.
\end{proof}

\section{Relation Between Embedding Distortion and Doubling Dimension}\label{sec:path-embedding-doubling-dimension}

We show that that every graph with a low-distortion embedding into the path also has small doubling dimension / doubling constant, as per the following lemma. Recall that the doubling constant of a metric space is the smallest value $\lambda$ such that, for any choice of radius $R \in \R$, every ball of radius $R$ can be covered with the union of $\lambda$ balls of radius $R/2$, and that the doubling dimension is given by $\log_2 \lambda$.

\begin{lemma}\label{lemma:path-embedding-vs-doubling-dim}
    Let $G$ be an undirected graph on $n$ vertices admitting an embedding into a path on $n$ vertices with distortion $\rho$ and let $\lambda$ be the doubling constant of $G$. Then:
    \[
        \lambda \leq \lceil 8\rho\rceil.
    \]
\end{lemma}

In contrast there exist graphs with constant doubling dimension that admit no embeddings into the unweighted path of distortion independent of $n$. For example, the 2D planar grid graph on $n$ vertices has constant doubling dimension / doubling constant, but a simple argument shows that every embedding of the 2D planar grid into the path has distortion $\Omega(\sqrt{n})$.

\begin{proof}[{\bf Proof of Lemma~\ref{lemma:path-embedding-vs-doubling-dim}}]
    Let $G=(V,E)$ be an undirected graph which embeds into $[n]$ with distortion $\rho$. Let $\tau$ be an embedding which achieves this distortion. For any $R > 0$, let $B_G(u, R) \subseteq V$ denote the ball in $G$ centered at $u$ of radius $R$, and let $B_{[n]}(\tau(u),R)$ denote the ball of radius $R$ in $[n]$ centered at $\tau(u)$. Let $\tau(S) \subseteq [n]$ denote the image of $S\subseteq V$ under $\tau$. For any radius $R$ and any $u\in V$, by the definition of the Lipschitz constant we can bound
    \[
        d_{[n]}(\tau(u),\tau(v)) \leq \Lipnorm{\tau}d_{G}(u,v) \leq \Lipnorm{\tau}R \qquad \forall v\in B_{G}(u,R)
    \]
    so $\tau(B_{G}(u,R))\subseteq B_{[n]}(\tau(u),\Lipnorm{\tau}R)$. In particular, for $S_1 \subseteq V$, $\tau(S_1)\subseteq S_2$ implies $S_1 \subseteq \tau^{-1}(S_2)$, so we conclude  $B_{G}(u,R)\subseteq \tau^{-1}(B_{[n]}(\tau(u),\Lipnorm{\tau}R))$.

    Consider $B_{[n]}(\tau(u),\Lipnorm{\tau}R)$. Fix $\epsilon$ and let $k_\epsilon$ denote the cardinality of an $\epsilon$-covering of  $B_{[n]}(\tau(u),\Lipnorm{\tau}R)$. Observe that for any $c, \epsilon >0$ and any $v\in [n]$, $B_{[n]}(v,c)$ admits an $\epsilon$-covering of cardinality at most $\lceil 2c/\epsilon\rceil$, so $k_\epsilon \leq \lceil 2\Lipnorm{\tau} R/\epsilon\rceil$. Let $\{x_1,\dots,x_{k_\epsilon}\}\subseteq [n]$ denote the centers of the covering balls. Given $B_{G}(u,R)\subseteq \tau^{-1}(B_{[n]}(\tau(u),\Lipnorm{\tau}R))$, we observe that
    \[
        B_{G}(u,R)\subseteq \bigcup_{i=1}^{k_\varepsilon} \tau^{-1}\left(B_{[n]}(x_i, \epsilon)\right).
    \]
    In particular, $\forall i\in [k_\epsilon]$, for any $x, y\in B_{[n]}(x_i, \epsilon)$, the definition of the Lipschitz constant implies
    \[
        d_{G}(\tau^{-1}(x),\tau^{-1}(y)) \leq \Lipnorm{\tau^{-1}} d_{[n]}(x,y) \leq \Lipnorm{\tau^{-1}}\cdot 2\epsilon.
    \]
    Thus $\text{diam}\left(\tau^{-1}\left(B_{[n]}(x_i, \epsilon)\right)\right) \leq 2\epsilon \Lipnorm{\tau^{-1}}$. In particular, this implies that $\forall i\in [k_\epsilon]$ such that $\tau^{-1}\left(B_{[n]}(x_i, \epsilon)\right)\neq \emptyset$, $\exists v_i\in V$ such that $\tau^{-1}\left(B_{[n]}(x_i, \epsilon)\right) \subseteq B_{G}(v_i, 2\epsilon\Lipnorm{\tau^{-1}})$. Thus
    \[
        B_{G}(u,R)\subseteq \bigcup_{i=1}^{k_\varepsilon} \tau^{-1}\left(B_{[n]}(x_i, \epsilon)\right) \subseteq \bigcup_{i=1}^{k_\varepsilon} B_{G}(v_i, 2\epsilon\Lipnorm{\tau^{-1}}).
    \]
    We have thus produced a covering of $B_{G}(u,R)$ using $k_\epsilon$ balls of radius $2\epsilon\Lipnorm{\tau^{-1}}$. We now choose $\epsilon$ so that $2\epsilon\Lipnorm{\tau^{-1}} = R/2$, namely let $\epsilon = R/(4\Lipnorm{\tau^{-1}})$. The cardinality of the covering is then
    \[
        k_{\epsilon} \leq \Bigg\lceil\frac{2\Lipnorm{\tau}R}{\epsilon}\Bigg\rceil= \Bigg\lceil 2\Lipnorm{\tau}R\cdot \frac{4\Lipnorm{\tau^{-1}}}{R}\Bigg\rceil = \lceil 8 \Lipnorm{\tau}\Lipnorm{\tau^{-1}} \rceil
    \]
    Using the fact that $\rho = \Lipnorm{\tau}\Lipnorm{\tau^{-1}}$ we conclude that the doubling constant of $G$ is at most $\lceil 8\rho \rceil$.
\end{proof}

\section{Experimental Details}\label{sec:further_experiments}

In this section, we provide a detailed descriptions of the experiments discussed in Section~\ref{sec:experiments} of the main body of the paper. We outline two sets of experiments: the first set of experiments is used to evaluate the performance of Algorithm~\ref{alg:greedy-l1-search} in the presence of absolute error, the second set is used to evaluate the performance of Algorithm~\ref{alg:vareps-known-search} in the presence of relative error. Both these sets of experiments focus on stochastic error.

All experiments in this section were run on a 2019 MacBook Pro with a 1.4 GHz Quad-Core Intel Core i5 Processor with 16 GB of RAM. No GPUs were used for this experiment.

\begin{figure}[H]
    \centering
    \includegraphics[scale=0.6]{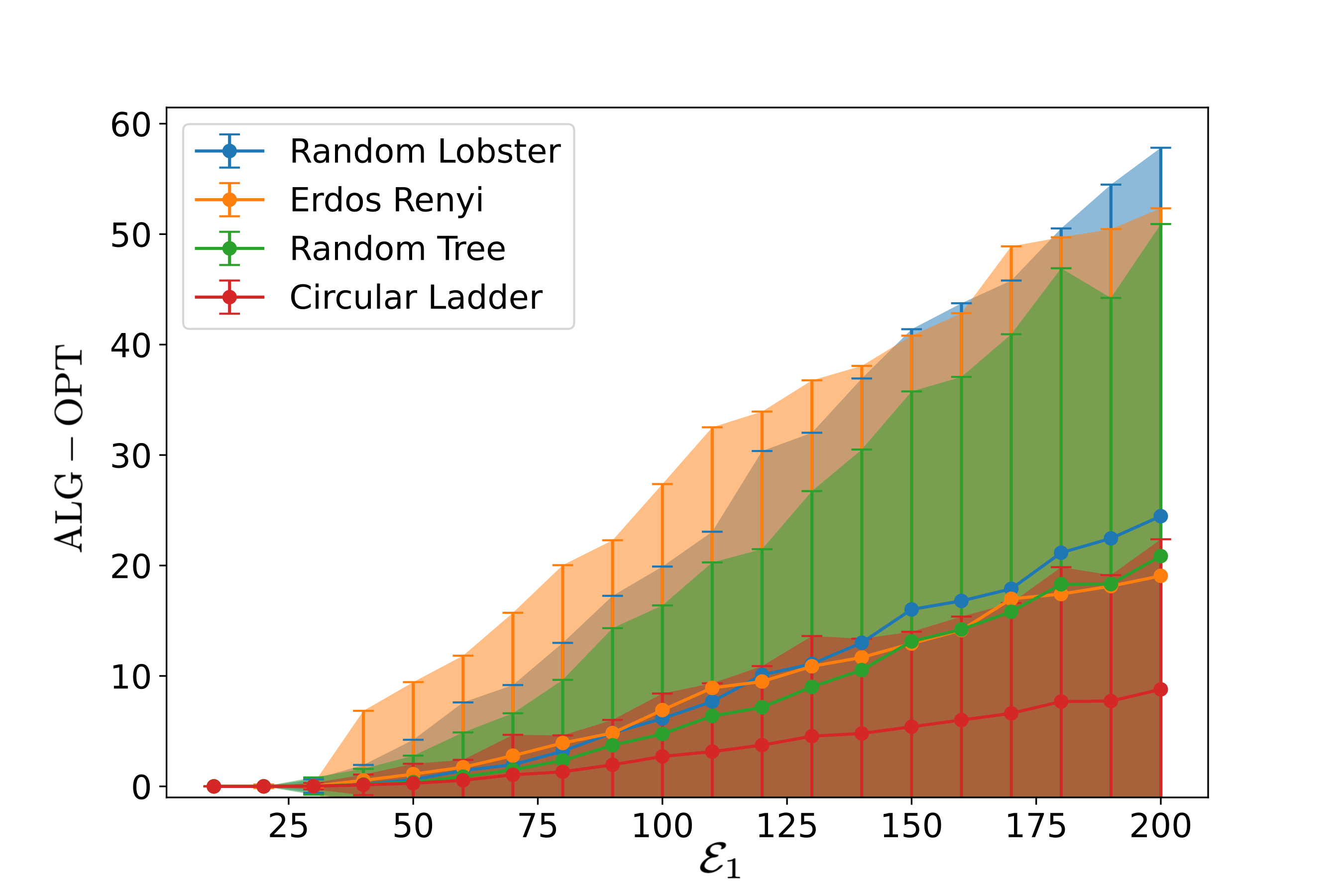}
    \caption{A larger rendering of the left subfigure in Figure~\ref{fig:main_body_experiments_absolute} in the main body of the paper.}
    \label{fig:supplementary_figure_absolute_error_random}
\end{figure}

\paragraph{Absolute Error} The first set of experiments corresponds to the left side of Figure~\ref{fig:main_body_experiments_absolute} (replicated above as Figure~\ref{fig:supplementary_figure_absolute_error_random}). Here, we generate an error vector $\vec{e} \in \R^n$ according to the following procedure: we fixed a value $\cE_1$ representing the total desired $\ell_1$-norm of the vector of errors, and then we sampled a vector $\vec{e}_{unsigned}$ uniformly at random from the scaled simplex with $\ell_1$-norm equal to $\cE_1$, i.e.:
\[
    \vec{e}_{unsigned} \sim \cE_1 \cdot \Delta_n \defeq \{ \vec{x} \in \R^n \mid \vec{x} \geq 0, \left\|{\vec{x}}\right\|_1 = \cE_1\}.
\]
We then assign a random sign to each entry of $\vec{e}_{unsigned}$ to obtain $\vec{e}$, this is done by multiplying each $\vec{e}_{unsigned}[v]$ by a Rademacher random variable $\sigma_v$. Fixing a graph $G$, the predictions at each vertex $v\in V$ are then given by $f(v) = d_G(v,g) + \sigma_v \cdot \vec{e}_{unsigned}[v]$. This is repeated over many instance graphs selected from four classes: Random Tree, Random Lobster, Erdos-R\'enyi and Circular Ladder (See paragraph \textbf{Graph Families} below). Whenever the family of graphs chosen is stochastic, as it is the case for all classes except for Circular Ladder, the graph is also resampled from its family at each iteration, so that the expectation is taken over the sampling of the graph topology as well as the random error. For each of these problem instances, we run Algorithm~\ref{alg:greedy-l1-search} and record the difference between the total distance $\alg$ travelled by the algorithm to find $g$, and the true shortest-path distance $\opt$ from the starting point to $g$, and we plot $\cE_1$ against it. We report mean and standard deviation of $\alg - \opt$ over 2000 independent trials.

\begin{figure}[H]
  \centering
  \includegraphics[scale=0.50,valign=t]{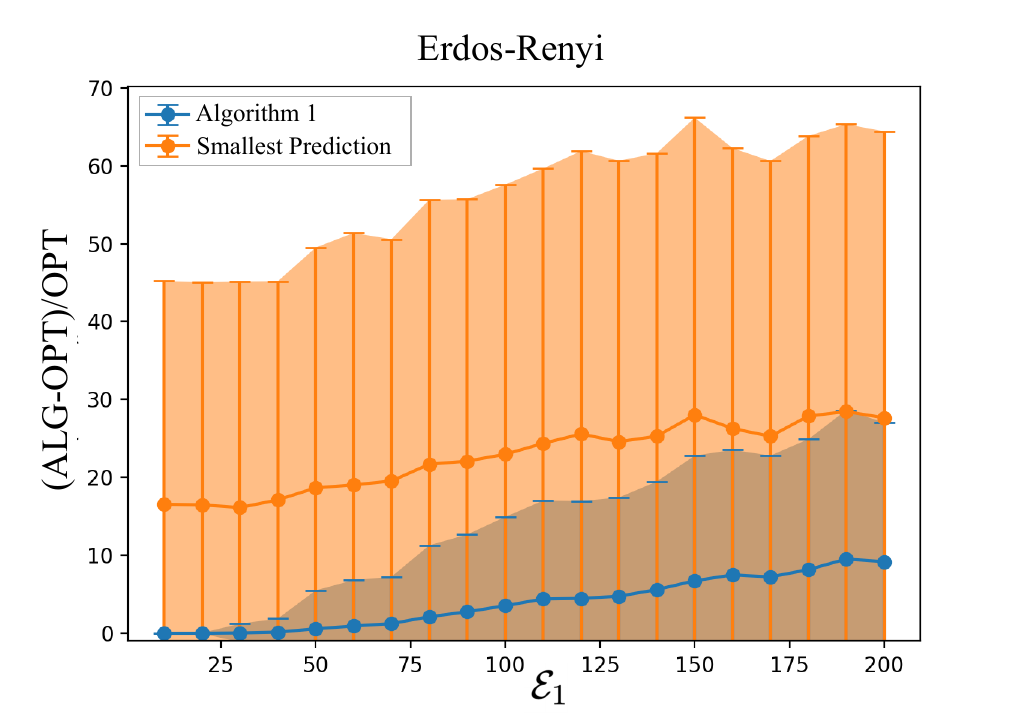}
  \includegraphics[scale=0.50,valign=t]{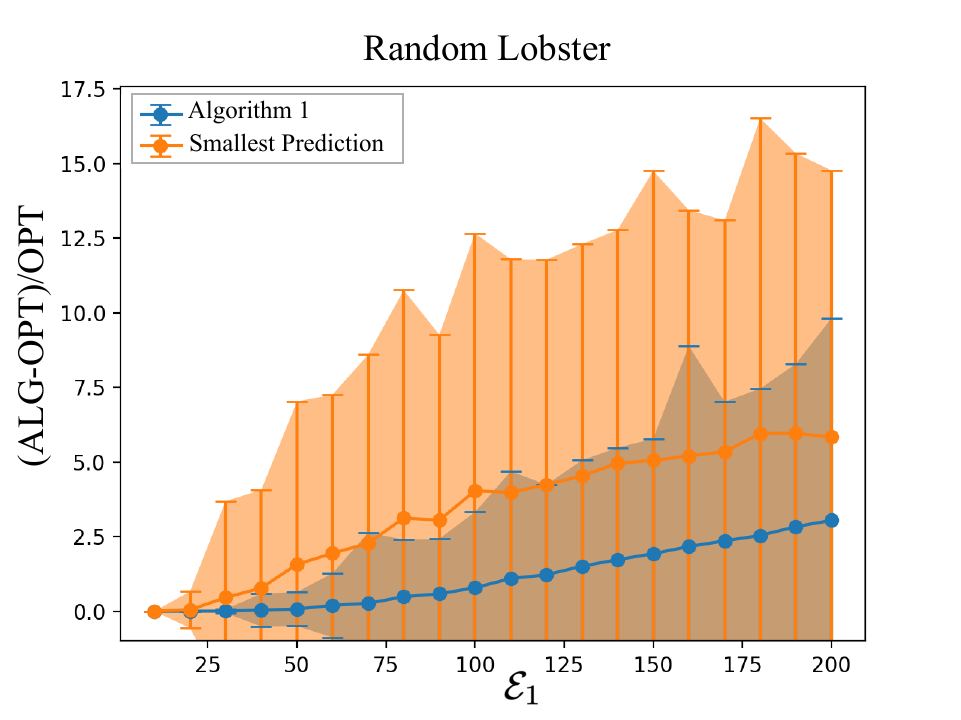}
  \includegraphics[scale=0.50,valign=t]{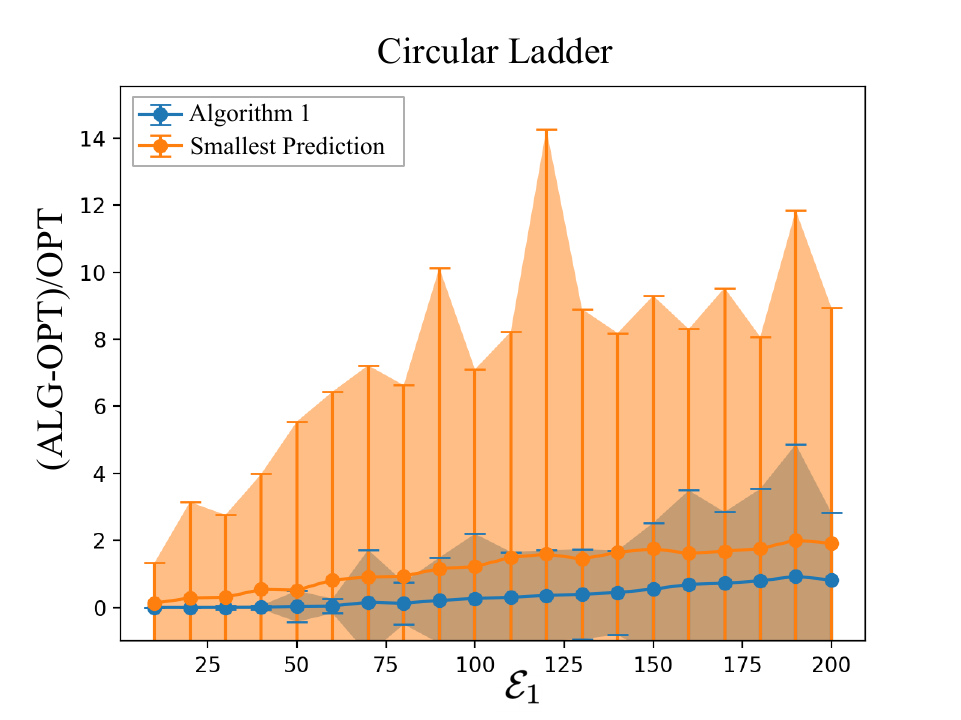}
  \caption{A comparison of the performance of \Cref{alg:greedy-l1-search} with the Smallest Prediction heuristic. Each subfigure represents one family of graphs: the top-left corresponds to random Erd\"{o}s-R\'enyi graphs, the top-right corresponds to Random Lobster, and the middle one at the bottom corresponds to circular ladder (See {\bf Graph Families} at the end of this section). In each figure, we plot the average and the standard deviation of the performance of \Cref{alg:greedy-l1-search} and that of the Smallest Prediction heuristic against the magnitude of the error vector $\cE_1$.
  }
    \label{fig:baseline_comparison_all}
\end{figure}

\paragraph{Comparison to Smallest Prediction Heuristic} We then compare the performance of \Cref{alg:greedy-l1-search} to the Smallest Prediction heuristic defined in \Cref{sec:experiments} (\Cref{fig:baseline_comparison_all}). Recall that in Smallest Prediction, at each iteration $i$, the agent travels to the an arbitrary vertex $v_i \in \argmin_{v\in V_{i-1}} f(v)$. We consider the same families of graphs as in the previous section. For each family, we compare the performance of our algorithm with that of Smallest Prediction for different values of the error magnitude $\cE_1$. The instances, including the errors, are generated like in the previous set of experiments. Performance is measured as the total distance travelled by the agent, minus the true distance $\opt$ from $r$ to $g$, as a fraction of $\opt$. Just like in the previous experiments, we run 2000 trials for every value of $\cE_1$ and report the average and standard deviation of the performance across those trials.

\newpage

\begin{figure}
    \centering
    \includegraphics[scale=0.6]{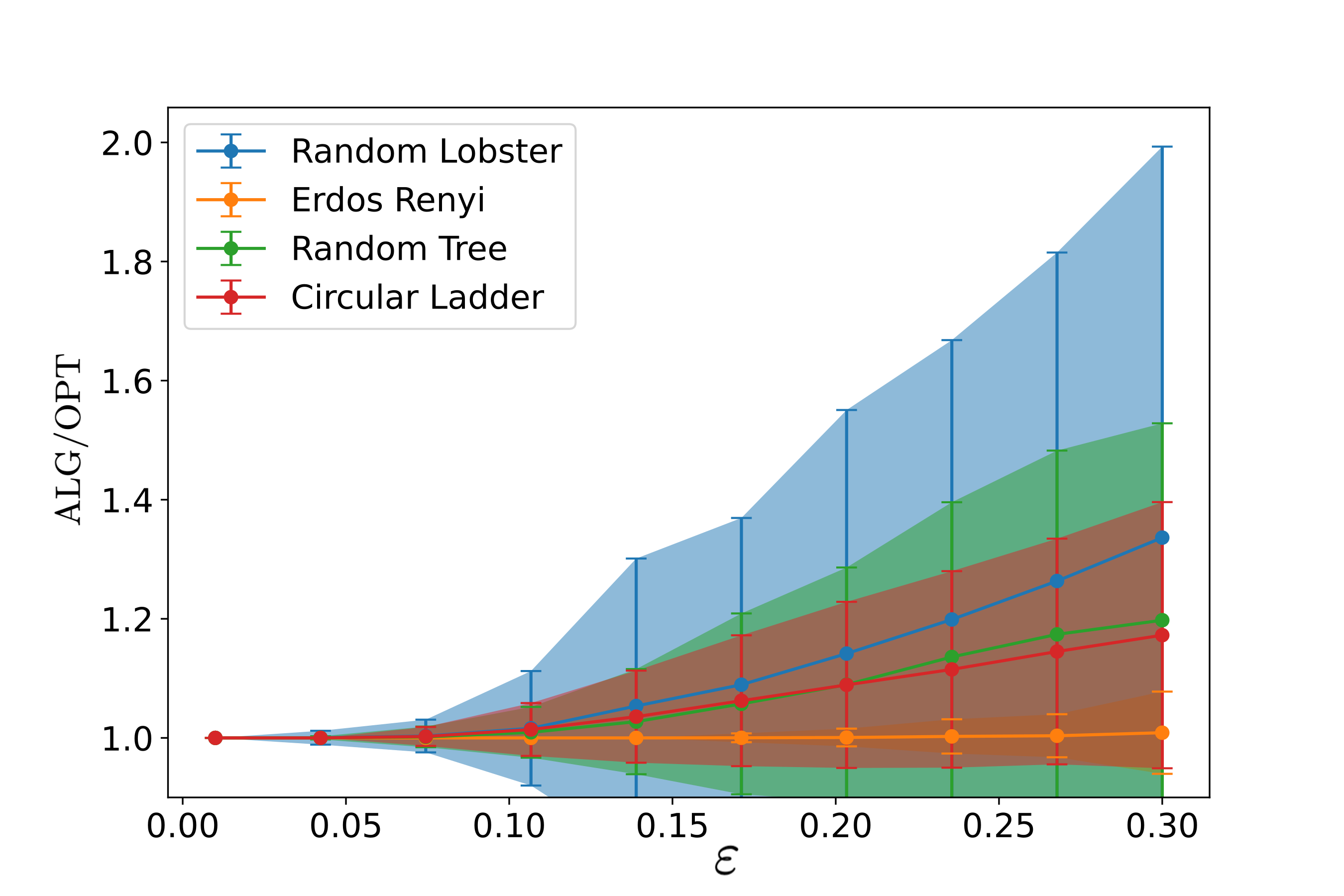}
    \caption{A larger rendering of the right subfigure of Figure~\ref{fig:main_body_experiments_absolute} in the main body of the paper.}
    \label{fig:supplementary_figure_relative_error_random}
\end{figure}

\paragraph{Relative Error} In the right subfigure of Figure~\ref{fig:main_body_experiments_absolute}, for each value of $\varepsilon$ the predictions are generated by setting $f(v) = (1+\varepsilon_v) \cdot d_G(v,g)$ where $\varepsilon_v$ is sampled from a Gaussian distribution with mean $0$ and standard deviation $\varepsilon/2$ conditioned on the event: $\varepsilon_v \in [-\varepsilon,\varepsilon]$. We run Algorithm~\ref{alg:weighted-vareps-unknown-search} and plot the value of the competitive ratio $\alg / \opt$ against the value of $\varepsilon$ for $\varepsilon\in [0,0.3]$, and report the mean and standard deviation incurred over 2000 independent trials.
\paragraph{Scaling with Number of Nodes} Finally, we plot the performance of \Cref{alg:weighted-vareps-unknown-search} as a function of $n$ (\Cref{fig:supplementary_figure_varying_num_nodes}). For this experiment we generate instances with different numbers of vertices and plot report the average and standard deviation of the respective empirical competitive ratios ($\alg / \opt$). We run $2000$ trials for each family of graphs and for each number $n\in \{50,100,500,1000\}$ of nodes. The error is generated as in the previous section with $\varepsilon=0.2$.

\begin{figure}
    \centering
    \includegraphics[scale=0.9]{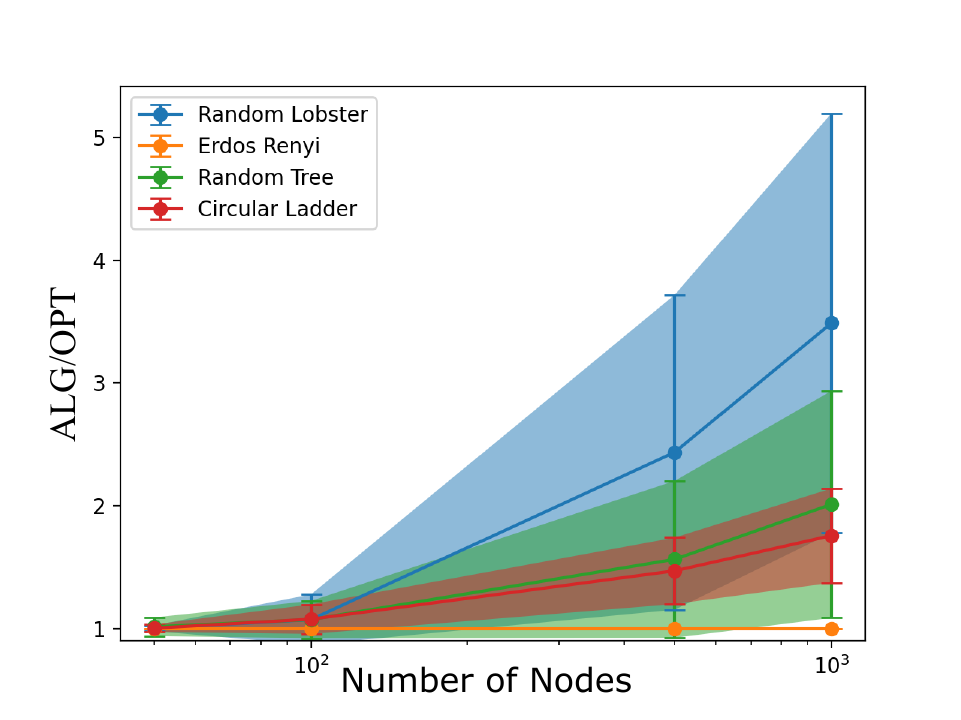}
    \caption{The empirical competitive ratio of \Cref{alg:weighted-vareps-unknown-search} for different graph families and for different values of $n$. On the $x$-axis: the number of vertices $n$ in the instance graphs considered. We consider $n=50,100,500,1000$. On the $y$-axis: the ratio between the distance travelled by the agent, and the true distance $\opt$ from $r$ to $g$ in $G$.}\label{fig:supplementary_figure_varying_num_nodes}
\end{figure}

\paragraph{Graph Families} In the above experiments we consider the four graph families described below. All the graphs considered are undirected and unweighted. In Figure~\ref{fig:main_body_experiments_absolute}, we sample the below graphs on $n=100$ nodes, and in Table~\ref{table:experiment_table}, we sample them on $n=300$ nodes.
\begin{itemize}
    \item Erd\"os-Renyi Random Graphs: Erd\"os-R\'enyi $G_{n,p}$ graphs~\citep{erdHos1960evolution} are a popular random graph model in the literature. We sample from the distribution of Erd\"os-R\'enyi graphs with $n$ nodes   and edge probability $p=0.1$, conditioned on the graph being connected;
    \item Random Trees: Trees are just connected acyclic graphs. We sample trees on $n$ vertices uniformly at random;
    \item Random Lobster Graphs: A lobster graph is a tree which becomes a caterpillar graph when its leaves are removed, we sample random lobster graphs on $n$ vertices; 
    \item Circular Ladder Graphs: A circular ladder graph is a graph obtained by gluing the endpoints of a ladder graph, i.e. it's a graph on vertices $\{1,... , 2k\}$ where the edges are of the form $(i,i+1)$ for $i=1,... , k-1$ and $i=k+1,... , 2k$, and $(i, k+i)$ for $i=1,... ,k$ as well as $(1,k)$, $(k+1,2k)$. We consider the circular ladder graph on $n$ vertices.
\end{itemize}

\vfill
\end{document}